\documentclass[final]{IEEEtran}
\IEEEoverridecommandlockouts
\usepackage{color}
\usepackage{amsmath}
\usepackage{amsfonts}
\usepackage{amssymb}
\usepackage{amscd}
\usepackage[dvips]{graphicx}
\usepackage{epsfig}
\usepackage{newlfont}
\usepackage{cite}
\usepackage{subfigure}
\newtheorem{theorem}{Theorem}

 \newtheorem{corollary}{Corollary}

\newcommand{\lb}{\left(}
\newcommand{\rb}{\right)}

\newcommand{\lcb}{\left\{}
\newcommand{\rcb}{\right\}}
\newcommand{\xbold}{\mathbf{x}}
\newcommand{\ybold}{\mathbf{y}}

\newcommand{\SNRt}{\text{SNR}}
\newcommand{\INRt}{\text{INR}}
\newcommand{\oone}{\mathcal{O}(1)}

\ifCLASSINFOpdf
\else
\fi

\normalsize

\begin{document}
\title{On the Secrecy Capacity Region of the $2$-user Z~Interference Channel with Unidirectional Transmitter Cooperation}
\author{\authorblockN{Parthajit~Mohapatra$^{*}$,~Chandra R. Murthy$^{\ddag}$, and ~Jemin~Lee$^*$}\\
	\authorblockA{$^*$iTrust, Centre for Research in Cyber Security, Singapore University of Technology and Design, Singapore\\
		$^\ddag$Department of ECE, Indian Institute of Science, Bangalore\\
		Email: \{parthajit,~jemin\_lee\}@sutd.edu.sg, cmurthy@ece.iisc.ernet.in}}
\maketitle
\vspace{-2cm}
\begin{abstract}
 In this work, the role of unidirectional limited rate transmitter cooperation is studied in the case of the $2$-user Z interference channel (Z-IC) with secrecy constraints at the receivers, on achieving two conflicting goals simultaneously: \emph{mitigating interference} and \emph{ensuring secrecy}. First, the problem is studied under the linear deterministic model.  The achievable schemes for the deterministic model use a
 fusion of cooperative precoding and transmission of a jamming signal. The optimality of the proposed scheme is established for the deterministic model
 for all possible parameter settings using the outer bounds derived by the authors in a previous work. Using the insights obtained from the deterministic model, a lower bound on the secrecy capacity region of the $2$-user Gaussian Z-IC are obtained. The achievable scheme in this case uses stochastic encoding in addition to cooperative precoding and transmission of a jamming signal. The secure sum generalized degrees of freedom (GDOF) is characterized and shown to be optimal for the weak/moderate interference regime. It is also shown that the secure sum capacity lies within $2$~bits/s/Hz of the outer bound for the weak/moderate interference regime
 for all values of the capacity of the cooperative link. Interestingly, in case of the deterministic model, it is found that there is no penalty on the capacity region of the Z-IC due to the secrecy constraints at the receivers in the weak/moderate interference regimes. Similarly, it is found that there is no loss in the secure sum GDOF for the Gaussian case due to the secrecy constraint at the receiver, in the weak/moderate interference regimes. The results highlight the importance of cooperation in
 facilitating secure communication over the Z-IC.
\end{abstract}
\section{Introduction}
The role of cooperation between the transmitters/receivers in interference limited scenarios has been studied extensively in the context of communication \emph{reliability}. 
However, the effect of the cooperation on communication \emph{secrecy} has not been well explored, and the ability to cooperate can have a very different effect on the achievable rates when there is a secrecy constraint (e.g., when the transmitted information should not be decodable at receivers except for the intended receiver) \cite{partha-arxiv-2014, partha-spawc-2013}. Therefore, this paper investigates the tension between the gain due to cooperation and the loss due to the secrecy constraints, on the rates achievable in an interference-limited communication system. In a system operating under secrecy constrains at receivers, the receivers cannot enhance their own rates by decoding and canceling the interference, since this does not preserve the communication secrecy. This leads to the following fundamental questions: 
(a) how much interference can be  mitigated through rate-limited transmitter cooperation, when there are secrecy constraints at receivers, and (b) what is the corresponding gain in the rate achieved by the cooperation between transmitters? Answering these questions helps in understanding the role of cooperation in managing interference and ensuring secrecy in multiuser scenarios.

The effect of transmitter cooperation on the  secrecy capacity is closely related to the underlying  channel model. The channel model considered in this paper is the Z-IC~\cite{liu-globecom-2004,liu-icst-2011}, which is one of the important information theoretic channel models.  In the Z-IC, only one of the two transmitters causes interference at the unintended receiver, which is also referred to as a \emph{partially} connected IC in~\cite{jafar-TIT-2014}. As a practical example, the Z-IC can model a 2-tier network, where the macro cell user is close to the edge of the femtocell  while the femtocell user is close to the femto base station (BS). 
Since the macro BS can typically support higher complexity transmission schemes, it could use the side information received from the femto BS to precode its data to improve its own rate and simultaneously ensure secrecy at the femtocell user. At the receivers, the macro cell user could experience significant interference from the femtocell BS, while the femtocell user receives little or no interference from the macro BS, leading to the Z-IC as the appropriate model for the system. Hence, answering the aforementioned questions in the context of the Z-IC can lead to useful insights in the 2-tier cellular network mentioned above.
\subsection{Prior work}
The IC has been studied extensively with and without secrecy constraints at the receivers under different settings \cite{etkin-TIT-2008, liu-TIT-2008,  lgamal2-TIT-2011}. However, the capacity region of the $2$-user Gaussian IC has remained an open problem, even without secrecy constraint, except for some specific cases like the strong interference regime and the very strong interference regime \cite{carleial-TIT-1975, sato-TIT-1981}. In \cite{liu-TIT-2008}, an achievable scheme using random binning is proposed for the discrete memoryless IC with secrecy constraints at receivers. A $K$-user Gaussian IC is considered in \cite{lgamal2-TIT-2011}, and the achievable scheme uses a combination of interference alignment along with precoding to ensure secrecy.

It has been shown that cooperation between the transmitters or receivers in case of IC can improve the overall performance of the system, when there is no secrecy constraint at the receiver \cite{wang-TIT-2011, wang2, vinod1}. However, the effect of cooperation on managing interference and ensuring secrecy in interference limited scenarios is not well understood. Some of the works in this direction can be found in \cite{partha-arxiv-2014, geng-globecom-2016}. It has been shown that, with cooperation, it is possible to achieve nonzero secrecy rate in most of the cases, \emph{even when the unintended receiver has a better channel compared to the legitimate receiver.} The effect of cooperation on the achievable rates for other communication models with secrecy constraints can be found in \cite{ekrem2, ekrem1, awan1, lgamal-TIT-2011}.

The Z-IC model has also been studied in existing literature with and without secrecy
constraints \cite{liu-globecom-2004, liu-TIT-2009, li-isit-2008}. In \cite{liu-globecom-2004}, lower bounds  on the capacity region of
the Gaussian Z-IC for the weak and moderate interference regimes are derived. In
\cite{liu-TIT-2009}, it is shown that superposition encoding with partial decoding is optimal for a certain class of Z-IC. In \cite{li-isit-2008}, the Z-IC model is considered with secrecy constraints at the receivers and achievable schemes are obtained for the deterministic and the Gaussian model in the weak/moderate interference regime. For the deterministic model, the secrecy capacity region is characterized. The role of cooperation in the Z-IC without the secrecy constraint has been investigated in~\cite{bagheri-arxiv-2010, lei-TIT-2012, do-allerton-2009}. In \cite{bagheri-arxiv-2010}, both the encoders can cooperate through noiseless
links with finite capacities and the sum capacity of the channel is
characterized to within $2$ bits of the outer bound. The role of receiver cooperation in Z-IC is investigated in \cite{lei-TIT-2012, do-allerton-2009}. However,
the role of limited transmitter cooperation in managing interference and ensuring secrecy in case of Z-IC has
not been investigated in the existing literature, and is therefore focus of this work.
\subsection{Contributions}
This work considers the $2$-user symmetric Z-IC with unidirectional transmitter cooperation in the form of a rate-limited link from
transmitter~$2$ (which causes interference) to transmitter~$1$ (which does not cause
interference), and with secrecy constraints at receivers.  The key challenge here is to devise techniques for simultaneously canceling interference and guaranteeing secrecy. First, the problem is solved under the deterministic approximation of the channel. By motivating from the results in the deterministic model, an achievable scheme is derived for the Gaussian channel model, which is applicable for all the interference regimes.

 One of the key techniques used in the achievable scheme for both the deterministic and Gaussian models is \emph{cooperative precoding} performed at transmitter~$1$, which cancels interference at receiver~$1$ and thereby simultaneously ensures secrecy. However,  the amount of the interference that can be canceled at the receiver is limited by the rate of the cooperative link. In the deterministic model, transmission of a jamming signal along with interference cancelation is required to achieve the capacity. On the other hand, the achievable scheme for the Gaussian model uses stochastic encoding in addition to cooperative precoding and transmission of a jamming signal.  The main contributions of the paper are summarized below:
\begin{enumerate}
 \item The achievable scheme for the deterministic model uses a careful combination of transmission of random bits and cooperative precoding to cancel the interference at the unintended receiver. The secrecy capacity region of the Z-IC can be characterized without sharing any common randomness between the transmitters for all values of the capacity of the cooperative link. This is in contract to the case of IC, where sharing common randomness between the transmitters through the cooperative link improves the performance of the achievable scheme compared to the case of sharing data bits only in some cases \cite{partha-arxiv-2014}. It is also shown that the capacity region of the deterministic Z-IC does not enlarge if the perfect secrecy constraint at the receiver is replaced with the weak or strong notion of secrecy.
  
  \item The achievable scheme for the Gaussian model uses a combination of stochastic encoding, interference cancelation and artificial noise transmission. The novelty in the achievable scheme lies in fusing stochastic encoding with interference cancelation.  In the weak/moderate interference regime, the secure sum generalized degrees of freedom (GDOF) is characterized and shown to be optimal. The secure sum capacity of the Z-IC is also shown to lie within $2$~bits/s/Hz of the outer bound in the weak/moderate interference regime for all possible values of the capacity of the cooperative link. 
  
  \item  The results on the secrecy capacity region of the $2$-user Z-IC without
  cooperation between the transmitters can be obtained as special case of the analysis for both the
  models. Note that, prior to this work, the capacity region of the Z-IC  for the deterministic model with secrecy constraints was
  not fully known, even for the non-cooperating case~\cite{li-isit-2008}.
  
\end{enumerate}

It is shown that limited-rate transmitter cooperation can greatly facilitate secure
communication over the Z-IC in weak/moderate and high interference regimes. In case
of the deterministic model, it is found, surprisingly, that there is no penalty on the capacity region of the Z-IC due to the secrecy constraints at the receivers in the weak/moderate interference regimes. Thus, the proposed scheme allows one to get secure communications for free. Similarly, it is found that there is no loss in the sum GDOF for the Gaussian case due to the secrecy constraint at the receiver, in the weak/moderate interference regimes. For the deterministic model, it is found that for every one bit increase in the capacity of the cooperative link, the secure sum rate can increase by one bit, in the weak, moderate and high interference regimes, until the sum rate is saturated by its maximum possible value. Part of this work has appeared in \cite{partha-isit-2015}.

\textit{Notation:} Lower case or upper case letters represent scalars, lower case
boldface letters represent vectors, and upper case boldface letters represent
matrices.

\textit{Organization:} Section~\ref{sec:system_model} presents the system model.
In Secs~\ref{sec:ach-determin} and \ref{sec:ach-gaussian}, the achievable schemes for the deterministic and Gaussian models are presented, respectively. Sec.~\ref{sec:approx-sec-capacity} presents the approximate secure sum capacity characterization in case of weak/moderate interference regime. In
Sec.~\ref{sec:results}, some numerical examples are presented to offer a
deeper insight into the bounds. Concluding remarks are offered in
Sec.~\ref{sec:conclusion}; and the proofs of the theorems are provided in the Appendices.
\section{System Model}\label{sec:system_model}
\begin{figure}[t]
\centering
\mbox{\subfigure[Gaussian model]{\includegraphics[width=1.5in, height=1.5in]{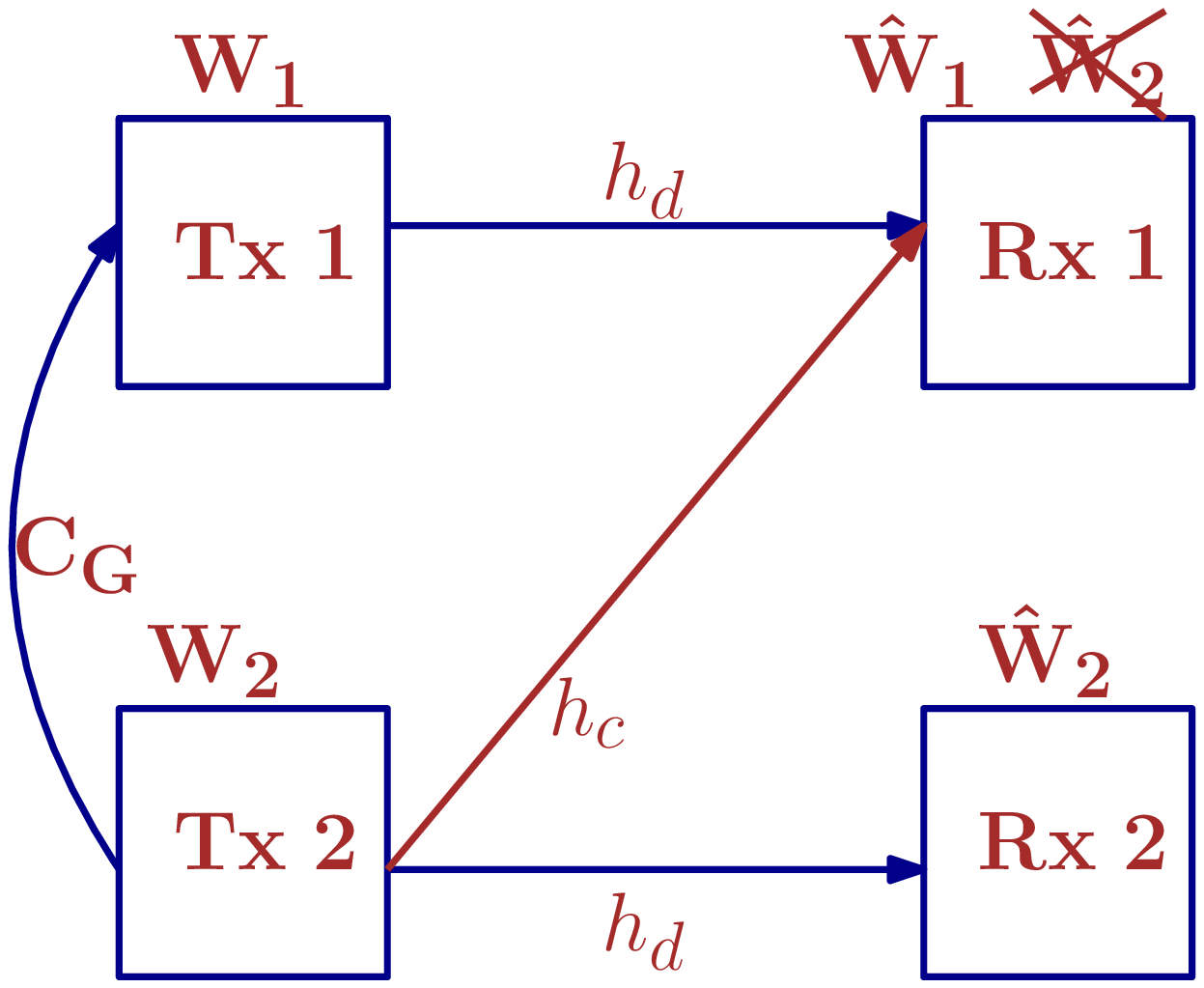} \label{fig:gaussian_model}}\quad
\subfigure[Deterministic model]{\includegraphics[width=1.5in, height=1.5in]{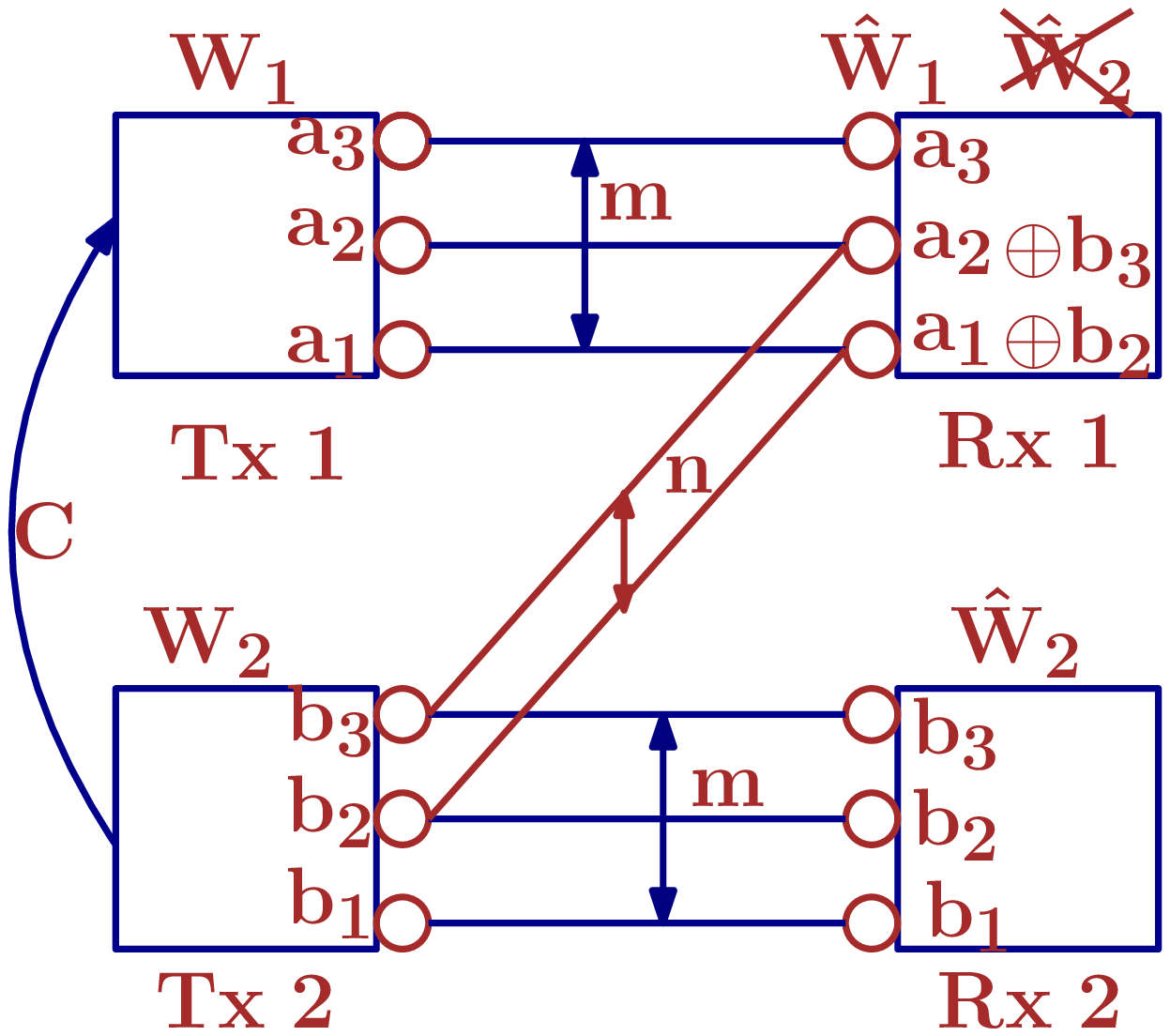} \label{fig:deterministic_model}}}
\caption[]{$2$-user Z-IC with unidirectional transmitter cooperation (from transmitter~$2$ to transmitter~$1$).}\label{fig:outersplit}
\end{figure}
Consider a $2$-user Gaussian symmetric Z-IC  with unidirectional and rate-limited transmitter
cooperation from transmitter~$2$ to $1$, as shown in Fig.~\ref{fig:gaussian_model}.\footnote{The model is
termed as symmetric as the links from transmitter~$1$ to receiver~$1$  and transmitter~$2$ to receiver~$2$
are of the same strength.} In the Z-IC, only one of the users (i.e., transmitter~$2$) causes interference to the unintended receiver (i.e., receiver~$1$). The received signal at
receiver~$i$, $\ybold_i$, is given by
\begin{align}
y_{1} = h_d x_1 + h_c x_{2} + z_1;  y_2 = h_dx_2 + z_2, \label{sysmodel1}
\end{align}
where $z_j$ $(j=1,2)$ is the additive white Gaussian noise, distributed as $\mathcal{N}(0,1)$. Here, $h_d$ and $h_c$ are the channel gains
of the direct and interfering links, respectively. The input signals ($x_i$)
are required to satisfy the power constraint: $E[|x_i|^2] \leq P$.  The transmitter~$2$ cooperates with transmitter~$1$ through a noiseless and secure link of finite rate denoted by $C_G$.

The equivalent deterministic model of \eqref{sysmodel1} at high SNR is given by~\cite{li-isit-2008,wang-TIT-2011}
\begin{align}
\mathbf{y}_{1} = \mathbf{D}^{q-m}\mathbf{x}_{1} \oplus \mathbf{D}^{q-n}\mathbf{x}_{2}; \  \mathbf{y}_{2} = \mathbf{D}^{q-m}\mathbf{x}_{2}, \label{sysmodel2}
\end{align}
where $\mathbf{x}_{1}$ ($\mathbf{x}_{2}$) is the binary input vector of the deterministic Z-IC from user~$1$
(user~$2$) of length $m$ ($\max\{m,n\}$); $\mathbf{y}_{1}$ $(\mathbf{y}_{2})$ is the binary output vector of
length $\max\{m,n\}$ ($m$);  $\mathbf{D}$ is a $q \times q$ downshift matrix with elements $d_{j',j''}=1$ if
$2 \leq j'=j''+1\leq q$ and $d_{j',j''}=0$ otherwise; and the operator $\oplus $ stands for  modulo-$2$ addition,
i.e., the \textsf{XOR} operation. The deterministic model is also shown in Fig.~\ref{fig:deterministic_model}.

The deterministic model is a first order approximation of a Gaussian channel, where all the signals are represented
by their binary expansions. Here, noise is modeled by truncation, and the superposition of signals at the
 receiver is modeled by \textit{modulo}~$2$ addition. Hence, the parameters $m$, $n$, and $C$ of the deterministic 
 model are related to the Gaussian symmetric Z-IC as $ m = (\lfloor 0.5 \log \text{\textsf{SNR}}\rfloor)^{+},\: n = (\lfloor 0.5 \log \text{\textsf{INR}}\rfloor)^{+},$ 
 and $C = \lfloor C_G\rfloor$. Note that the
 notation followed for the deterministic model is the same as that presented in \cite{wang-TIT-2011}. The
 bits $a_{i} \in \mathcal{F}_{2}$ and $b_i \in \mathcal{F}_{2}$ denote the  information bits of transmitters $1$ and $2$, respectively, sent
 on the $i^{\text{th}}$ level, with the levels numbered starting from the bottom-most entry. 

The transmitter~$i$ has a message $W_{i}$, which should be decodable at the intended receiver~$i$,
but needs to be kept secret from the other, i.e., the unintended receiver $j$ ($j \neq i$), and this is termed as the \emph{secrecy constraint}. Note that, for the Z-IC, the message $W_1$ is secure as there is no link from transmitter~$1$ to receiver~$2$. Hence, the goal is to ensure that $W_2$ is not decodable at receiver~$1$. The encoding at  
transmitter~$1$ should satisfy the causality constraint, i.e., it cannot depend 
on the signal to be sent over the cooperative link in the future. The signal sent over the cooperative link from
transmitter~$2$ to transmitter~$1$ is represented by $\mathbf{v}_{21}$. It is
also assumed that the transmitters trust each other completely and they do not deviate from the agreed
schemes, for both the models.
 
For the deterministic model, the encoded message at transmitter~$1$ is a function of its own data bits, the bits received through the
cooperative link, and possibly some random bits, whereas the encoded message at transmitter~$2$ is independent
of the other user's data bits.  The bits transmitted on the different levels of the deterministic
model are chosen to be equiprobable Bernoulli distributed, denoted by
 $\mathcal{B}(\frac{1}{2})$. The decoding is based on solving the linear equation in \eqref{sysmodel2}
at each receiver. For secrecy, it is required to satisfy the perfect secrecy constraint, i.e., $I(W_{i}; \mathbf{y}_{j}) = 0, i,j \in \{1,2\} \text{ and } i \neq j$
in the case of the deterministic model~\cite{shannon-bell-1949}. In the later part of the sequel, it is shown that replacing the perfect secrecy constraint at receiver with the strong or weak secrecy constraint  does not enlarge the capacity region of the deterministic model.

In the Gaussian case, the details of the encoding and decoding
schemes can be found in Sec.~\ref{sec:ach-gaussian}. For
the Gaussian model, the notion of weak secrecy is considered, i.e., $\frac{1}{N} I(W_2; \ybold_1^N) \rightarrow 0$ as $N \rightarrow \infty$, where $N$ corresponds to the block length~\cite{wyner-bell-1975}. 

The following interference regimes are considered: weak/moderate interference regime $(0 \leq \alpha \leq 1)$, high interference regime $(1 < \alpha \leq 2)$
and very high interference regime $(\alpha > 2)$, where, with slight abuse of notation $\alpha \triangleq
\frac{n}{m}$ is used for the deterministic model and $\alpha \triangleq \frac{\log \INRt}{\log \SNRt}$ is used for the
Gaussian model. The quantity $\alpha$ captures the amount of coupling between the
signal and interference.

\section{Linear Deterministic Z-IC: Achievable Schemes}\label{sec:ach-determin}
When there is a high capacity cooperative link
from transmitter~$2$ to transmitter~$1$, the interference caused at receiver~$1$ by transmitter~$2$ can be completely canceled 
by using the signal received from transmitter~$2$ via the cooperative link at transmitter~$1$. This cancelation
of interference offers two benefits: it improves the achievable rate, and also ensures secrecy, since the signal sent by transmitter~$2$ is no longer decodable at receiver~$1$.  When the capacity of the
cooperative link is not sufficiently high, it is not possible to design the precoding to completely eliminate the interference caused by transmitter~$2$
at receiver~$1$. In this case, the transmission of random bits (i.e., transmission of artificial noise \cite{yates1, tang-TIT-2011}) by transmitter~$1$ can ensure secrecy 
of the data bits sent by transmitter~$2$ at receiver~$1$, in turn enabling transmitter~$2$ to achieve a higher secure rate of communication. Thus, the achievable scheme proposed below uses 
a carefully designed combination of interference
cancelation and transmission of random bits depending on the capacity of the cooperative link $C$ and the value of $\alpha$. In the following, it is shown that the corner points of the outer bounds on the secrecy capacity region of the deterministic model  \cite{partha-isit-2015} are achievable for different interference regimes. Hence, the achievable results stated in Theorems~\ref{th:theorem-det-weak}-\ref{th:theorem-det-veryhigh} correspond to the secrecy capacity region of the deterministic Z-IC. 
\subsection{Weak/moderate interference regime $(0 \leq \alpha \leq 1)$}\label{sec:ach-det-weakmod}
\begin{theorem}\label{th:theorem-det-weak}
In the weak/moderate interference regime, i.e., $0 \leq \alpha \leq 1$, the secrecy capacity region of the $2$-user deterministic Z-IC with unidirectional and rate-limited transmitter cooperation is
\begin{align}
& R_1 \leq m, R_2 \leq m, R_1 + R_2 \leq 2m-n + C. \label{eq:th-weakmod-ib1}
\end{align}
\end{theorem}
\begin{proof}
In this regime, using Theorem~$1$ in \cite{partha-isit-2015}, the secrecy capacity region of the deterministic Z-IC is upper bounded as  $R_1 \leq m$, $R_2 \leq m$, and $R_1 + R_2 \leq 2m-n + C$. Thus, the outer bound  is characterized by four corner rate pairs 
$(R_1, R_2)$ corresponding to $(m,0)$, $(0,m)$, $(m,m-n+C)$ and $(m-n+C,m)$. In the following, an achievable scheme 
is proposed to achieve these corner points, thus, achieving the capacity of the deterministic Z-IC in
the weak/moderate interference regime.
\subsubsection{Case 1 $(R_1,R_2)=(m,0)$ and $(R_1,R_2)=(0,m)$}
To achieve the point $(m,0)$, transmitter~$1$ sends $m$ data bits on the levels $[1:m]$ and transmitter~$2$ 
remains silent. To achieve the corner point $(0,m)$, transmitter~$2$ sends data bits on the levels $[1:m]$. As 
the data bits sent on the levels $[m-n+1:m]$ are received at receiver~$1$, transmitter~$1$ sends random bits 
generated from $\mathcal{B}(\frac{1}{2})$ distribution on the levels $[1:n]$ to ensure secrecy of the data bits of 
transmitter~$2$ at receiver~$1$.
\subsubsection{Case 2 $(R_1,R_2)=(m,m-n+C)$}
As data bits sent by transmitter~$1$ are not received
at receiver~$2$, it can send $m$ data bits securely. Transmitter~$2$ can send at least $m-n$ data bits 
securely as the links corresponding to the levels $[1:m-n]$ are not present at receiver~$1$. To transmit 
at the higher levels $[m-n+1:m-n+C]$, transmitter~$2$ shares $C$ cooperative bits with the  transmitter 
and the transmitter~$1$ precodes these cooperative data bits with its own data bits and transmits on the levels $[1:C]$.
Transmitter~$2$ also sends the data bits shared with the transmitter~$1$ on the levels $[m-n+1:m-n+C]$ so that the data bits of transmitter~$2$ can be canceled at receiver~$1$. This scheme not only cancels the interference, but at the same time ensures secrecy. Hence, transmitter~$2$ can achieve a rate of $m-n+C$, while transmitter~$1$ can achieve a rate of $m$.
\subsubsection{Case 3 $(R_1,R_2)=(m-n+C,m)$}
In contrast to the previous case, the achievable 
scheme in this case uses transmission of random bits in addition to interference cancelation. Transmitter~$2$ sends data 
bits on the levels $[1:m]$ and shares $C$ data bits $[b_{m-n+1}:b_{m-n+C}]$ with transmitter~$1$. Transmitter~$1$ 
precodes these cooperative data bits with its own data bits and sends these precoded data bits on the levels $[1:C]$. 
On the levels $[C+1:n]$, transmitter~$1$ sends random bits to keep the data of transmitter~$2$ confidential to 
receiver~$1$.\footnote{When $C=n$, the interference caused by transmitter~$2$ can be completely 
	eliminated at receiver~$1$, and it is not required to send random bits from transmitter~$1$.} On the remaining 
levels $[n+1:m]$, transmitter~$1$ sends its own data bits. As transmitter~$2$ does not cause any interference to the data bits sent on these levels by transmitter~$1$, receiver~$1$ can decode these data bits. Hence, transmitter~$1$ achieves a
rate of $m-n+C$, while transmitter~$2$ achieves a rate of $m$.

The achievable schemes for the two cases (Case~2~and~3) are also illustrated in Fig.~\ref{fig:weakach3}. 
\begin{figure}[t]
	\centering
	\mbox{\subfigure[][$(R_1,R_2)=(5,3)$] {\includegraphics[width=1.7in,height=1.8in]{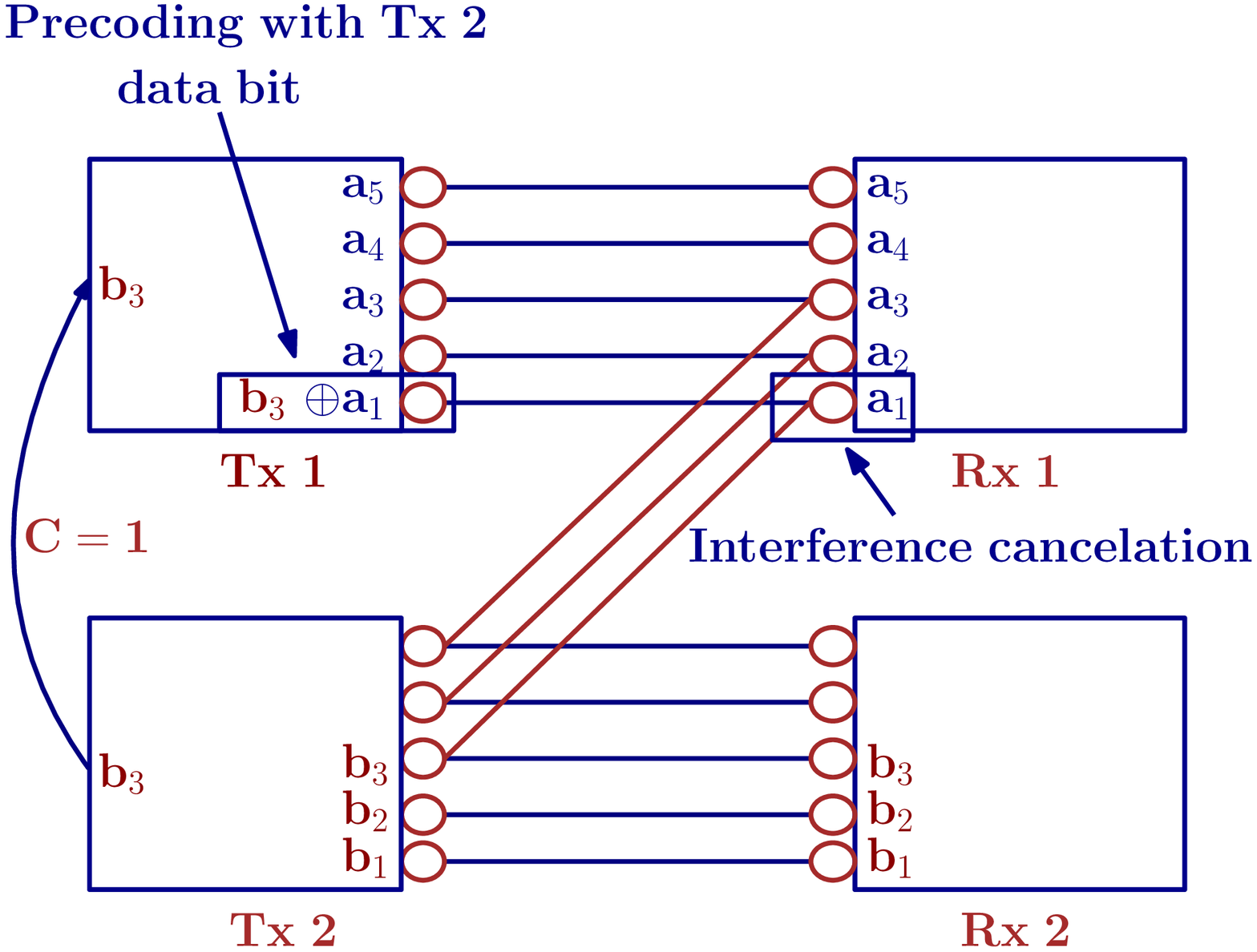}\label{fig:weakach1}} \quad
		\subfigure[][$(R_1,R_2)=(3,5)$] {\includegraphics[width=1.7in,height=1.8in]{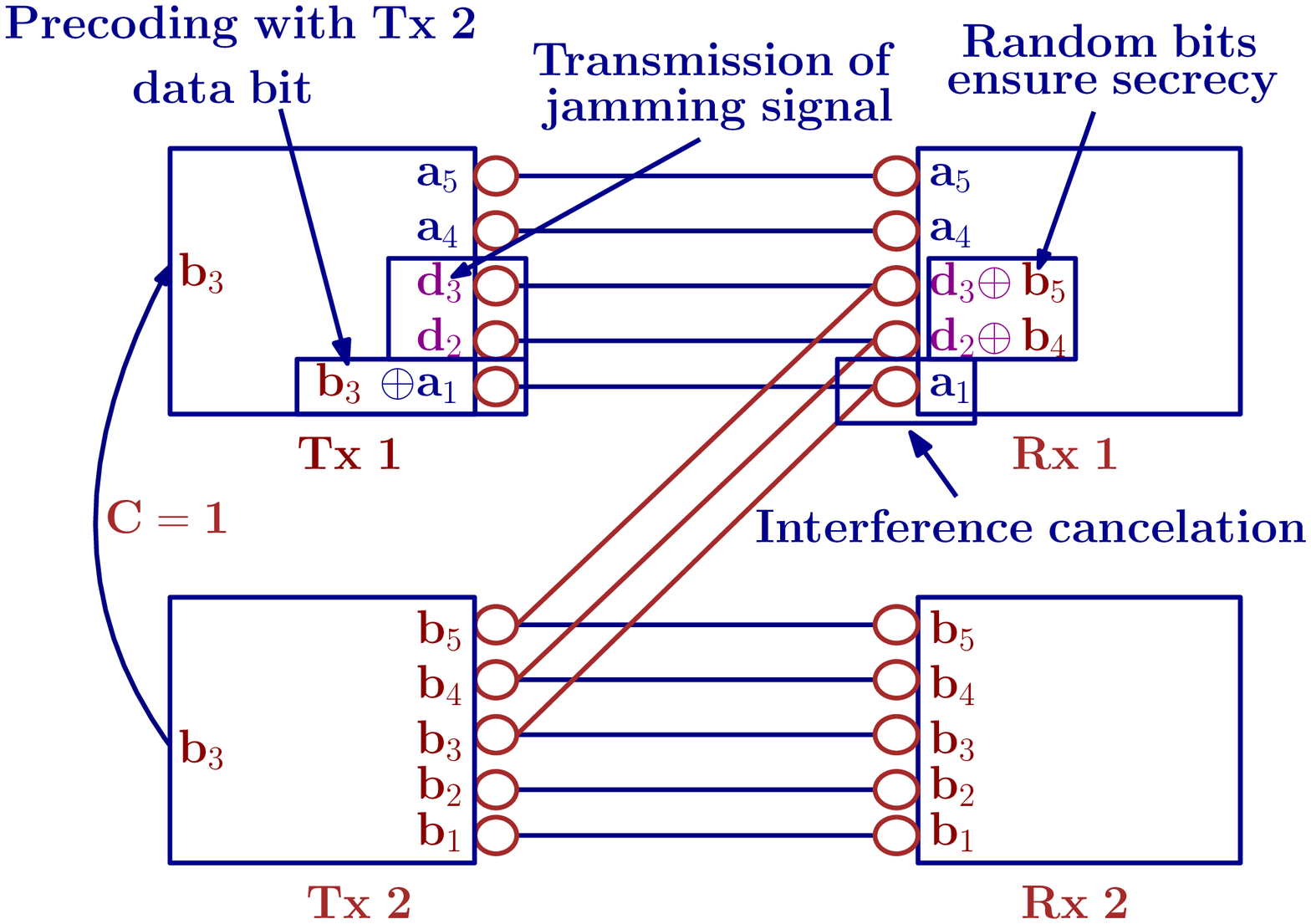}\label{fig:weakach2}}} \\
	\caption[]{Deterministic Z-IC with $m=5$, $n=3$ and $C=1$.}\label{fig:weakach3}
\end{figure}
\end{proof}
\textit{Remarks:}
\begin{itemize}
	\item The derivation of the outer bound in \cite{partha-isit-2015} does not use the secrecy constraint at the receiver. The proposed schemes can achieve the four corner points of the outer bound, and  hence, the secrecy constraints at the receivers do not result in any penalty on the capacity region. Thus, the capacity region of the deterministic Z-IC is characterized with and without secrecy constraints for all the values of $C$.
	\item When $0 < \alpha \leq 1$, both the users can achieve the maximum rate of $m$ simultaneously if $C \geq m$, otherwise they cannot.
\end{itemize}
\subsection{High interference regime $(1 < \alpha < 2)$}\label{sec:ach-det-high-intf}
\begin{theorem}\label{th:theorem-det-high}
	In the high interference regime, i.e., $1 < \alpha < 2$, the secrecy capacity region of the $2$-user deterministic Z-IC with unidirectional and rate-limited transmitter cooperation is
	\begin{align}
	& R_1 \leq m, R_2 \leq 2m -n,  R_1 + R_2 \leq m + C. \label{eq:th-weakmod-ib2}
	\end{align}
\end{theorem}
\begin{proof}
In this regime, using Theorem~$2$ in \cite{partha-isit-2015}, the secrecy capacity region of the deterministic Z-IC is upper bounded as $R_1 \leq m$, $R_2 \leq 2m-n$, and $R_1 + R_2 \leq m + C$. This outer bound is characterized by four corner rate pairs $(R_1, R_2)$ corresponding 
to $(m,0)$, $(0,2m-n)$, $(m,C)$ and $(C+ n-m,2m-n)$. The achievability of the first two corner points is trivial. 
Note that, for achieving these two corner points, sharing of data from transmitter~$2$ to transmitter~$1$ is not required. 
In the following, achievability of the remaining two corner points is shown. 
\subsubsection{Case 1 $(R_1,R_2)=(m,C)$} To achieve this corner point, transmitter~$1$ sends data bits on 
the levels $[1:m]$ and hence, it can achieve a secrecy rate of $m$. The transmitter~$2$ shares $[b_1:b_C]$ to 
transmitter~$1$ through the cooperative link. Transmitter~$1$ precodes (XOR) the cooperative data bits 
received from transmitter~$2$ with its own data bits, and transmits them on the levels corresponding to $[n-m+1:n-m+C]$. 
Precoding in this way not only cancels the interference caused by the data bits transmitted on the levels $[n-m+1:n-m+C]$ at transmitter~$2$, but also ensures secrecy. Hence, transmitter~$2$ can send $C$ data bits securely.
\begin{figure}
	\centering
	\mbox{\subfigure[][$(R_1,R_2)=(4,1)$]{\includegraphics[width=1.7in, height=1.8in]{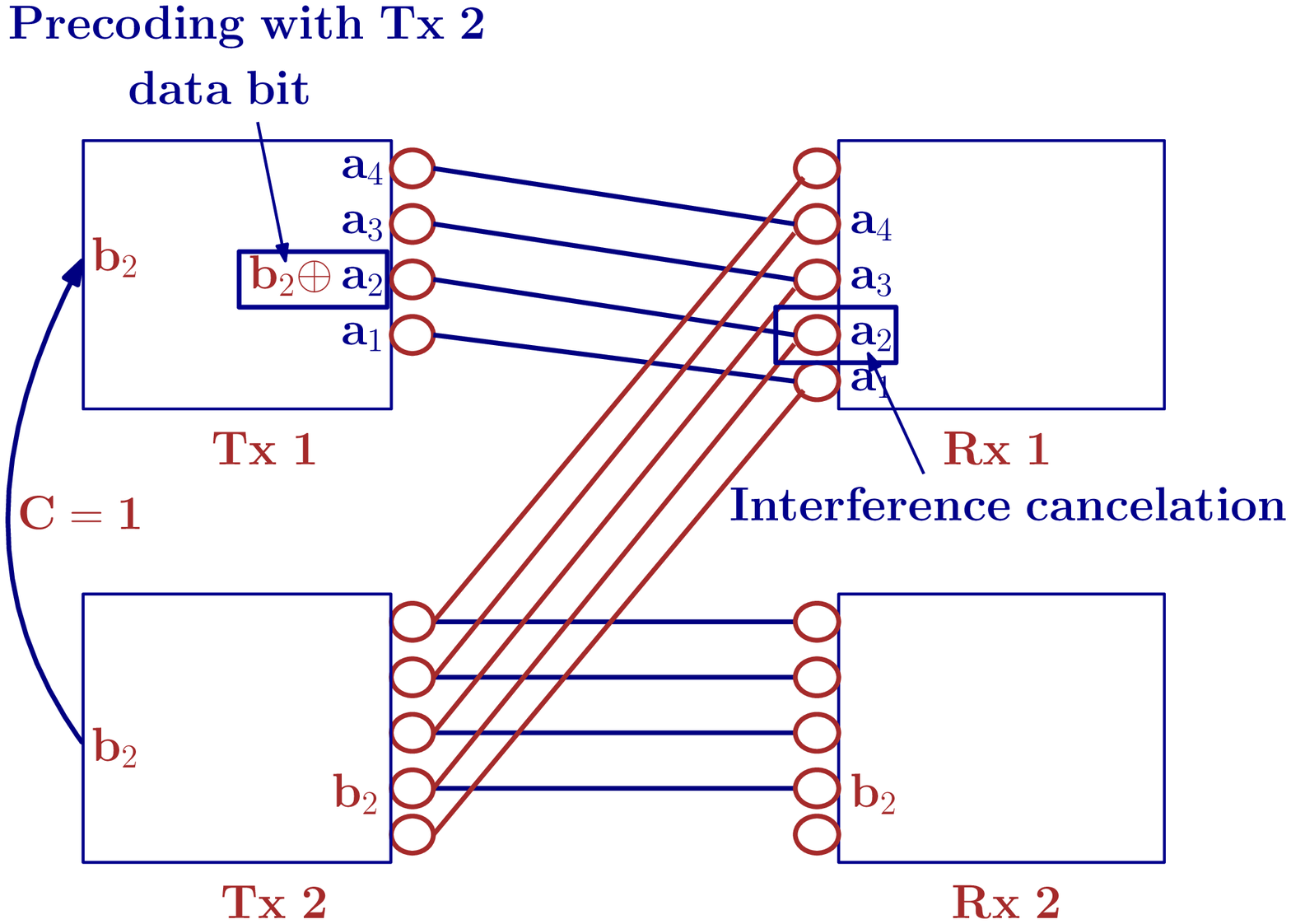}\label{fig:highach1}} \quad
		\subfigure[][$(R_1,R_2)=(2,3)$] {\includegraphics[width=1.7in,height=1.8in]{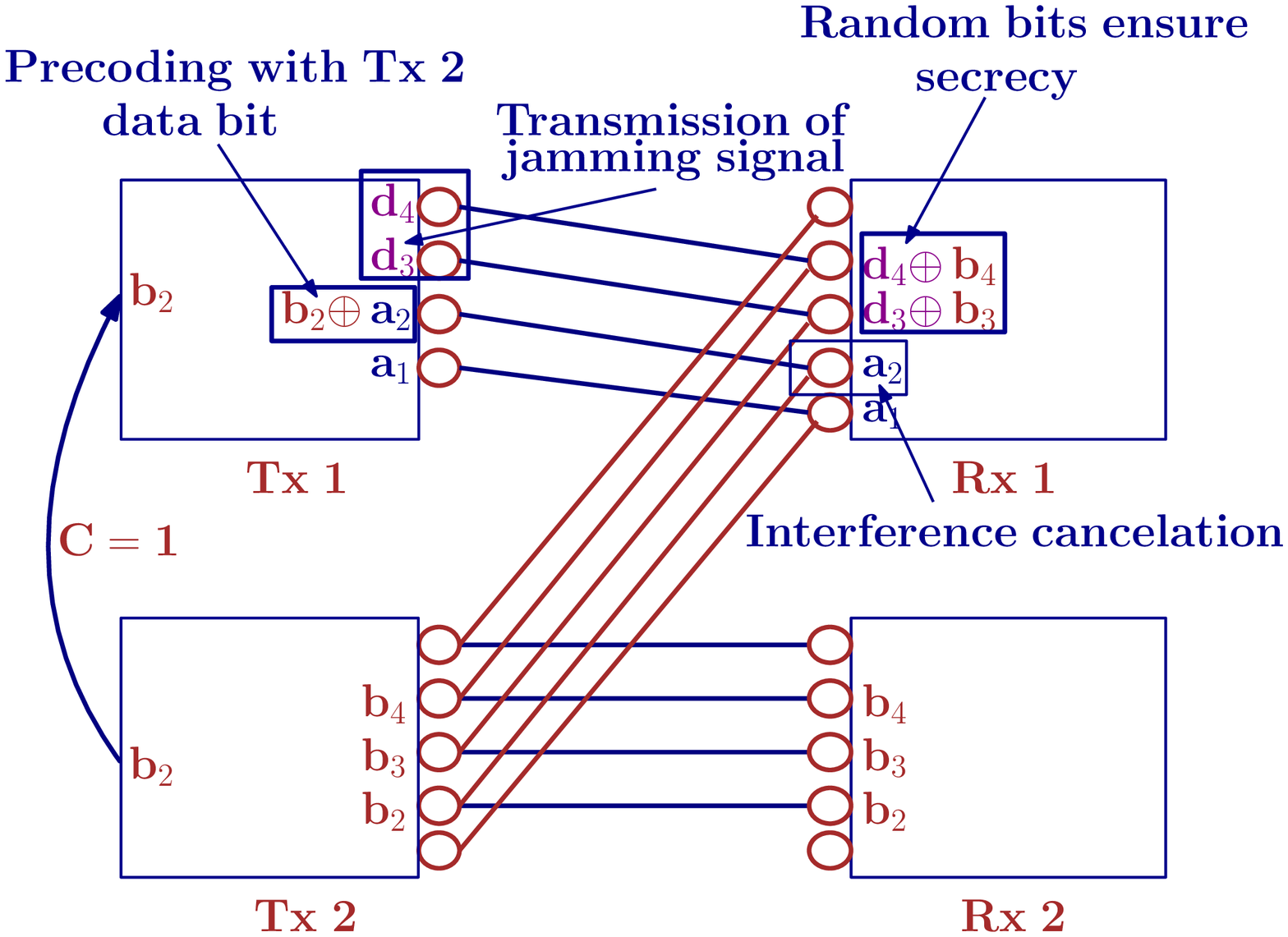}\label{fig:highach2}}} \\
	\caption[]{Deterministic Z-IC with $m=4$, $n=5$ and $C=1$.}\label{fig:highach3}
\end{figure}
\subsubsection{Case 2 $(R_1,R_2)=(C+ n-m,2m-n)$} In this case, as the links corresponding to levels $[1:n-m]$
are only present from transmitter~$2$ to receiver~$1$, transmitter~$2$ does not send any data bits on these
levels. Hence, transmitter~$1$ can send data bits on the levels corresponding to $[1:n-m]$ and these data bits
will be received without any interference at receiver~$1$. When $C>0$, the achievable scheme uses interference
cancelation in addition to transmission random bits. Bits transmitted on the levels $[n-m+1:n-m+C]$ will cause
interference at receiver~$1$. The interference can be eliminated by precoding the data bits on levels $[n-m+1:n-m+C]$
at transmitter~$1$ with the data bits of transmitter~$2$, received through cooperation. Transmitter~$2$ can 
also send data bits securely on the levels corresponding to $[n-m+C+1:m]$ with the help of transmission of 
random bits by transmitter~$1$ on the levels $[n-m+C+1:m]$. Transmitter~$2$ remains silent on the remaining 
top levels, as transmission of data bits on these levels violates the secrecy condition. Hence, the corner point $(C+n-m,2m-n)$ is achievable by the proposed scheme.

The achievable schemes for the two cases are illustrated in Fig.~\ref{fig:highach3}. Note that, to 
achieve the maximum possible rate of $2m-n$ for user~$2$, it is not required to use the cooperative link, and transmitter~$1$ 
sends random (jamming) bits on the levels $[1:m]$. 
As all the corner points are achievable, any points between these corner points is also achievable by time sharing. 
Hence, the capacity region of the Z-IC in the high 
interference regime is characterized.
\end{proof}
\textit{Remarks:} 
\begin{itemize}
	\item When $C=0$ and $1 < \alpha < 2$, if user~$1$ achieves the maximum  rate of $m$, then user~$2$ cannot achieve any nonzero secrecy rate. This is in contrast to the weak/moderate 
	interference case, where user~$1$ achieves the maximum rate of $m$,  while user~$2$  achieves the rate of $m-n$ even without cooperation.
	\item When $1 < \alpha < 2$ and $C \geq 2m-n$, transmitters~$1$ and $2$ can simultaneously achieve the maximum rates of $m$ and $2m-n$, respectively. 
	\item In general, the principle behind the achievable schemes to achieve the corner points $(m, m-n + C)$ and $(m, C)$ in the weak/moderate and high interference regimes, respectively, is precoding of data bits at transmitter~$1$ using the data bits of transmitter~$2$ received on the cooperative link to cancel interference and ensure secrecy. On the other hand, the achievability of the corner points $(m-n + C, m)$ and $(n-m +C, 2m-n)$ in the weak/moderate and high interference regimes, respectively, requires transmission of random bits by transmitter~$1$ to ensure that the signal from transmitter~$2$ remains secure, in addition to precoding data bits received from transmitter~$2$ with its own data bits. 
\end{itemize}
\subsection{Very high interference regime $(\alpha \geq 2)$}
\begin{theorem}\label{th:theorem-det-veryhigh}
In the very high interference regime, i.e., $\alpha \geq 2$, the secrecy capacity of the $2$-user deterministic Z-IC with unidirectional and rate-limited transmitter cooperation is
	\begin{align}
	& R_1 \leq m, R_2 = 0. \label{eq:th-weakmod-ib3}
	\end{align}
\end{theorem}
\begin{proof}
The outer bound on the rate of user~$2$ in Theorem~3 \cite{partha-isit-2015} shows that user~$2$ cannot achieve any nonzero secrecy rate irrespective of the capacity of the cooperative link. Thus, transmitter~$1$ can send data bits on the levels $[1:m]$, while transmitter~$2$ remains silent. This characterizes the capacity of the deterministic Z-IC in the very high interference regime.
\end{proof}

Interestingly, it turns out that the capacity region of the Z-IC does not change
if the perfect secrecy constraint at the receiver is replaced with the strong or the weak notion of secrecy. This 
result is stated in the following Theorem.
\begin{theorem}
	The secrecy capacity region of the deterministic Z-IC with unidirectional transmitter cooperation satisfies the
	following
	\begin{align}
	\mathcal{C}^{\text{perfect}} = \mathcal{C}^{\text{strong}} = \mathcal{C}^{\text{weak}}, \label{eq:coor-capacity}
	\end{align}
	where $\mathcal{C}^{\text{perfect}}$, $\mathcal{C}^{\text{strong}}$ and $\mathcal{C}^{\text{weak}}$
	correspond to the capacity regions of the $2$-user deterministic Z-IC with unidirectional transmitter cooperation guaranteeing the perfect, strong and weak secrecy constraints at the receivers, respectively.
\end{theorem}
\begin{proof}
Any communication scheme satisfying the perfect secrecy condition  will automatically satisfy the strong and weak secrecy conditions. Similarly, a communication scheme satisfying strong secrecy will automatically satisfy the weak secrecy condition. Hence, the following holds
	\begin{align}
	\mathcal{C}^{\text{perfect}} \subseteq \mathcal{C}^{\text{strong}} \subseteq 
	\mathcal{C}^{\text{weak}} \subseteq \mathcal{C}_{\text{outer}}^{\text{weak}},
	\label{eq:coor-capacity1}
	\end{align}
	where $\mathcal{C}_{\text{outer}}^{\text{weak}}$ corresponds to the outer bound on the capacity region of the Z-IC with unidirectional transmitter cooperation and weak secrecy constraints at the receivers. The achievable results in Sec.~\ref{sec:ach-determin} are obtained under the perfect secrecy constraints at the
	receivers. On the other hand, it is not difficult to show that the outer bounds on the capacity region in \cite{partha-isit-2015} do not change if perfect secrecy constraint is replaced with weak secrecy constraint.\footnote{This can be shown by using  $\frac{1}{N}I(W_i ; \ybold_j^N) \leq \epsilon$, $i \neq j$,  (weak secrecy) as a measure of secrecy in the derivation of the outer bounds, instead of $I(W_i, \ybold_j) = 0$ (perfect secrecy).} As the achievable rate regions $(\text{i.e.}, \mathcal{C}^{\text{perfect}})$ match with the outer bounds on the
	capacity region $(\text{i.e.}, \mathcal{C}_{\text{outer}}^{\text{weak}})$, the relation in (\ref{eq:coor-capacity}) holds.
\end{proof}

\section{Gaussian Z-IC: Achievable Scheme}\label{sec:ach-gaussian}
For the Gaussian case, a unified achievable scheme is proposed, which is applicable in the weak, moderate
and high interference regimes. The achievable scheme is based on the cooperative 
precoding performed at the transmitters to cancel the interference  at the 
unintended receiver along with stochastic encoding and transmission of 
artificial noise. When the capacity of the cooperative link is not sufficiently 
high, it is not possible to share the entire message of transmitter~$2$ with receiver~$1$ through the cooperative link. Hence, the interference caused at receiver~$1$ by transmitter~$2$ cannot be completely eliminated. Thus, stochastic encoding performed at transmitter~$2$ and artificial noise transmission by transmitter~$1$ can 
provide additional randomness to increase the secrecy rate of user~$2$. 

The achievable scheme is inspired by the approaches used for the deterministic model in Secs.~\ref{sec:ach-det-weakmod} and \ref{sec:ach-det-high-intf}. However, the extension of schemes proposed for the 
deterministic model to that for the Gaussian model is non-trivial.  This is primarily because, in the deterministic case, noise is modeled by truncation and superposition of signals modeled as XOR operation does not account for the carry over of bits across levels, in contrast to the Gaussian model. 

The achievable scheme for the deterministic model is extended to the Gaussian model as follows. Since there is no cooperative link from transmitter~$1$ to transmitter~$2$, transmitter~$1$ cannot share its
message with transmitter~$2$ for cooperation. The message of transmitter~$1$ intended to receiver~$1$ is
inherently secure, as there is no link from transmitter~$1$ to receiver~$2$. This translates to having non-cooperative 
private message $w_{p1} \in \mathcal{W}_{p1} = \{1,2,\ldots,2^{NR_1}\}$ at transmitter~$1$, and 
for each message, it transmits a codeword from a Gaussian codebook 
of  size $2^{NR_{1}}$. Next, for the transmission of data by transmitter~$2$, recall that, in the deterministic 
case, the data bits sent by transmitter~$2$ on the lower levels $[1:m-n]$ are inherently secure in the weak/moderate 
interference regime (Sec.~\ref{sec:ach-det-weakmod}). To enable secure transmission of data bits on the higher 
levels (specifically, levels $[m-n+1 : m]$ in the weak/moderate interference regime and levels $[n-m+1:n]$ in the high interference regime), 
transmitter 2 needs the assistance of transmitter 1. That is, transmitter 1 needs to precode the data bits 
received through the cooperative link, or needs to send a jamming signal so that the other
user's data bits remain undecodable at  receiver~$1$. To translate this scheme to the Gaussian case,
the message at transmitter~$2$ is split into two parts: a non-cooperative private message
$w_{p2} \in \mathcal{W}_{p2}=\{1, 2, \ldots, 2^{NR_{p2}}\}$ and a cooperative private message
$w_{cp2} \in \mathcal{W}_{cp2} = \{1, 2, \ldots, 2^{NR_{cp2}}\}$. Transmitter~$2$ encodes
the non-cooperative private message into $\xbold_{p2}^N$ using stochastic encoding. A stochastic encoder
is specified by a matrix of conditional probability $f_{p2}(x_{p2,k}|w_{p2})$, where $x_{p2, k} \in \mathcal{X}_{p2}$ 
and $w_{p2} \in \mathcal{W}_{p2}$.

For the cooperative private message, transmitters~$1$ and $2$ precode the message $w_{cp2}$ cooperatively such that the codeword carrying the cooperative private message is completely canceled at the non-intended receiver. This cooperative precoding also helps ensure secrecy for the cooperative private message. The details of the encoding and decoding process of the achievable scheme are presented in the following subsection. 
\subsection{Encoding and decoding}
For the non-cooperative private part, transmitter $1$ generates a codebook $\mathcal{C}_{p1}$ containing
$2^{NR_1}$ i.i.d. sequences of length $N$ and its entries are i.i.d. random 
variables from $\mathcal{N}(0, P_{p1})$. Transmitter~$2$ generates two codebooks as 
follows. For the non-cooperative private message, it generates a 
codebook $\mathcal{C}_{cp2}$ containing $2^{N(R_{p2} + R_{p2}')}$ codewords of 
length $N$. The entries of the codebook are drawn  at random from $\mathcal{N}(0, 
P_{p2})$. The $2^{N(R_{p2} + R_{p2}')}$ codewords in the codebook $\mathcal{C}_{p2}$ are randomly grouped into $2^{NR_{p2}}$
bins, with each bin containing $2^{NR_{p2}'}$ codewords. Any codeword in $\mathcal{C}_{p2}$ is indexed as
$\xbold_{p2}^N(w_{p2},w_{p2}')$ for $w_{p2} \in \mathcal{W}_{p2}$ and $w_{p2}' \in \mathcal{W}_{p2}' = \{1,2,\ldots,2^{NR_{p2}'}\}$.
To send $w_{p2}$, transmitter $2$ selects  $w_{p2}'$ uniformly random from the set $\mathcal{W}_{p2}'$  and transmits 
the codeword $\xbold_{p2}^N(w_{p2},w_{p2}')$. For the cooperative private 
message, transmitter~$2$ generates a codebook $\mathcal{C}_{cp2}$ consisting of 
$2^{NR_{cp2}}$ i.i.d. sequences of length $N$. The entries of the codebook are 
chosen at random from $\mathcal{N}(0, P_{cp2})$. This codebook is made available 
at transmitter~$1$. 

To send a message $(w_{p2}, w_{cp2})$, transmitter~$2$ superimposes the cooperative codeword
$\xbold_{cp2}(w_{p2})$ with the non-cooperative codeword  $\xbold_{p2}^N(w_{p2},w_{p2}')$ as 
\begin{align}
\xbold_2^N(w_{p2}, w_{p2}', w_{cp2}) = \xbold_{p2}^N(w_{p2},w_{p2}') + h_d \xbold_{cp2}^N(w_{cp2}). \label{eq:weakint5}
\end{align}
The following power constraint is required to be satisfied at transmitter~$2$:
\begin{align}
P_{p2} + h_d^2 P_{cp2} \leq P,  \label{eq:weakint6c}
\end{align}
where $P_{p2}$ and $P_{cp2}$ are parameters to be chosen later.

Transmitter~$1$ performs precoding as mentioned in (\ref{eq:weakint6}), so that the codeword carrying the
cooperative private message of transmitter~$2$ is canceled at receiver~$1$. This is termed as \textit{cooperative precoding}.
Transmitter~$1$ also adds artificial noise $(\xbold_{a1}^N)$ generated from a Gaussian distribution to increase the achievable secrecy rate
for transmitter~$2$. Thus, transmitter~$1$ sends
\begin{align}
\xbold_1^N(w_{p1}, w_{cp2}) = \xbold_{p1}^N(w_{p1}) - h_c \xbold_{cp2}^N(w_{cp2}) + \xbold_{a1}^N. \label{eq:weakint6}
\end{align}
The following power constraint is required to be satisfied at transmitter~$1$:
\begin{align}
P_{p1} + h_c^2 P_{cp2} + P_{a1} \leq P, \label{eq:weakint6d}
\end{align}
where $P_{p1}$ and $P_{a1}$ are parameters to be chosen later.

The decoding at receivers is performed as follows. Receiver~$1$ looks for a unique index 
$\hat{w}_{p1}$ such that $(\ybold_1^N, \xbold_1^N(\hat{w}_{p1}))$ is jointly typical. Receiver~$2$ looks for a 
unique tuple $(\hat{w}_{p2}, \hat{w}_{p2}',\hat{w}_{cp2})$ such that $(\ybold_2^N, \xbold_{p2}^N(\hat{w}_{p2}, \hat{w}_{p2}'), \xbold_{cp2}^N(\hat{w}_{cp2}))$ 
is jointly typical. Decoding errors at the receivers can occur in one of two ways. First, the receiver may not be able to find any 
codeword that is jointly typical with the received sequence. Second, a wrong codeword is jointly typical with 
the received sequence.

Based on the above encoding  and decoding strategy, the following theorem gives a lower bound on the secrecy 
capacity of the Z-IC with unidirectional transmitter cooperation.
\begin{theorem}\label{th:theorem-ib-weak}
	For the Gaussian Z-IC with unidirectional transmitter cooperation and
	secrecy constraints at the receivers, the achievable rate region is given by
	\begin{align}
	R_1 & \leq I(\xbold_{p1};\ybold_1), \nonumber \\
	R_2 & \leq \min\lcb I(\xbold_{p2}, \xbold_{cp2};\ybold_2), I(\xbold_{p2};\ybold_2|\xbold_{cp2}) + \min\{C_G, \right. \nonumber \\
	&\left. I(\xbold_{cp2};\ybold_2|\xbold_{p2})\} \rcb  - R_{p2}',  \label{eq:weakint6a}
	\end{align}
	where $R_{p2}' = I(\xbold_{p2};\ybold_1|\xbold_{p1})$.
\end{theorem}
\begin{proof}
	See Appendix~\ref{sec:append-ZIC-inner}.
\end{proof}
\textit{Remarks:}
\begin{enumerate}
	\item The term $R_{p2}'$ in Theorem~\ref{th:theorem-ib-weak} accounts for the rate sacrificed by transmitter~$2$ in confusing receiver~$1$ to keep the non-cooperative message of transmitter~$2$ secret. As the  capacity
	of the cooperative link increases, the loss in rate due to the stochastic encoding decreases, as more power can be assigned to the cooperative private message.
	\item When $C_G=0$ and $\alpha \geq 1$, the transmission of artificial noise by transmitter~$1$ is required along with stochastic
	encoding for user~$2$ to achieve a non-zero secrecy rate. 
\end{enumerate}

By evaluating the mutual information terms in (\ref{eq:weakint6a}) and 
taking convex closure of the union of set of regions obtained over different codebook parameters $(P_{p1}, P_{a1}, P_{p2}, P_{cp2})$, 
the following lower bound on the secrecy capacity region is obtained. 
\begin{corollary}\label{th:corr-ib-weak}
	Using the result in Theorem~\ref{th:theorem-ib-weak}, the following 
	rate region is achievable
	\begin{align}
	\mathcal{R}_s \triangleq \text{convex closure of} \lcb \bigcup_{ 0 \leq (\theta_i, \beta_i, \lambda_i) \leq 1,\:i = 1, 2} \mathcal{R}_{\text{Z-IC}}^s (\theta_i, \beta_i, \lambda_i)\rcb, \label{eq:corrach0}
	\end{align}
	where 
	\begin{align}
	\mathcal{R}_{\text{Z-IC}}^s \triangleq & \Bigg\{ (R_1, R_2): R_1 \geq 0, R_2 \geq 0,  \nonumber \\
	& \quad R_1 \leq 0.5\log\left(1 + \frac{h_d^2P_{p1}}{1 + h_d^2P_{a1} + h_c^2 P_{p2}}\right), \nonumber \\
	&  \quad R_2 \leq 0.5\log(1 + h_d^2P_{p2} + h_d^4P_{cp2}) - R_{p2}', \nonumber \\
	& \quad R_2 \leq 0.5\log(1 + h_d^2P_{p2}) + \min\{C_G,  \nonumber \\
	& \quad 0.5\log(1 + h_d^4P_{cp2})\} - R_{p2}'\Bigg\}, 
	\end{align}
	where $R_{p2}' \triangleq 0.5\log\left(1 + \frac{h_c^2P_{p2}}{1 + h_d^2P_{a1}}\right)$, $P_{cp2} \triangleq \frac{\lambda_2}{(\lambda_1 + \lambda_2)h_d^2}P_2$, 
	$P_{p2} \triangleq \frac{\lambda_1}{\lambda_1 + \lambda_2}P_2$,  
	$P_{p1} \triangleq \frac{\theta_1}{\theta_1 + \theta_2}P'$, $P_{a1} \triangleq \frac{\theta_2}{\theta_1 + \theta_2}P'$, $P' \triangleq (P_1 - h_c^2P_{cp2})^{+}$, 
	$P_1 \triangleq \beta_1 P$, and $P_2 \triangleq \beta_2 P$.
\end{corollary}
\begin{proof}
	See Appendix~\ref{sec:append-ZIC-inner-corr}.
\end{proof}
\textit{Remarks:}
\begin{enumerate}
	\item In Corollary~\ref{th:corr-ib-weak}, the parameter $\beta_i$ $(0 \leq \beta_i \leq 1)$ acts as a power control parameter for transmitter~$i$ $(i=1,2)$. The parameters
	$\theta_i$ and $\lambda_i$ act as rate splitting parameters for transmitter~$i$. 
	
	\item When $C=0$ $(\text{or } C_G=0)$, the system reduces to the 2-user Z-IC (Gaussian Z-IC) without cooperation, which was studied in \cite{li-isit-2008}. The achievable results in Theorem~$2$ (Theorem~$3$)  in \cite{li-isit-2008} can be obtained as a special case of achievable results for the deterministic model (Gaussian model) in 
	Theorem~\ref{th:theorem-det-weak} (Theorem~\ref{th:theorem-ib-weak}), by setting $C = 0$ ($C_G = 0$) and $0 \leq \alpha \leq 1$. Note that, for both the deterministic and Gaussian models, achievable schemes  on the secrecy capacity region have not been addressed in the literature for the high interference regime $(\alpha > 1)$, even when $C = 0$ ($C_G = 0$).
\end{enumerate}
\section{Approximate secure sum capacity characterization of the Gaussian Z-IC in the weak/moderate interference regime}\label{sec:approx-sec-capacity}
\subsection{Secure sum generalized degrees of freedom (GDOF)}
As mentioned earlier, the capacity region for many multiuser scenarios has remained an open problem, even without secrecy constraints at the receivers. Due to this, there has been an active research interest in approximate characterizations of the capacity. In this context, the notion of \emph{generalized degrees of freedom (GDOF)} has been used as a proxy for the capacity at high SNR and INR, for the IC, without the secrecy constraint~\cite{etkin-TIT-2008}. A natural extension of this to the secure sum GDOF is given by
\begin{align}
d_{\text{sum}}(\kappa, \gamma) = \lim\limits_{\SNRt \rightarrow \infty} \frac{C_{\text{sum}}(\SNRt, \INRt)}{0.5\log \SNRt}, \label{eq:gdof1} 
\end{align}
where $\kappa = \lim\limits_{\SNRt \rightarrow \infty} \frac{\log \INRt}{\log \SNRt}$, $\gamma = \lim\limits_{\SNRt \rightarrow \infty} \frac{C_G}{0.5\log\SNRt}$ and $C_{\text{sum}}$ is the secure sum capacity of the 2-user Gaussian Z-IC with unidirectional transmitter cooperation. To characterize the sum GDOF, $h_d=1$ is assumed without loss of generality, and the following power allocation is used. 
\begin{align}
P_{p1} = \frac{P}{2}, P_{p2} = \frac{1}{h_c^2},  P_{cp2} = \frac{1}{2}\lb P - \frac{1}{h_c^2}\rb \text{ and } P_{a1}=0. \label{eq:gap2}
\end{align}
It is also assumed that $h_c^2 P > 1$, so that the above power allocation is always feasible. The motivation for this power allocation is as follows. The power for the message of transmitter~$1$ is set as $\frac{P}{2}$ to ensure that user~$1$ achieves the maximum GDOF of $1$. Recall that, in the weak/moderate interference regime, transmitter~$2$ can send data bits securely on the lower levels $[1:m-n]$, as the links corresponding to these levels are not present at receiver~$1$. In other words, the data bits transmitted on the lower levels $[1:m-n]$ of transmitter~$2$ are received at or below the noise floor of receiver~$1$. Hence, in the Gaussian case, the power for the non-cooperative private message is chosen such that it is received at the noise floor of the receiver~$1$. Due to this power allocation, the loss in rate of user~$2$ due to stochastic encoding is $R_{p1}' = 0.5$ bits/s/Hz. Hence, the loss in achievable secrecy rate due to stochastic encoding does not scale with  SNR and INR.  The cooperative private message of transmitter~$2$ is assigned with a power of  $\frac{1}{2}\lb P - \frac{1}{h_c^2}\rb$. 

In the following theorem, the secure sum GDOF is characterized using the power allocation in (\ref{eq:gap2}) for all values of $C_G$ in the weak/moderate interference regime. 
\begin{theorem}\label{th:sumgdof}
	The optimal secure sum GDOF of the $2$ user Gaussian Z-IC with unidirectional transmitter cooperation in the weak/moderate interference regime is
	\begin{align}
	d_{\text{sum}}(\kappa, \gamma) = \min\lcb 2, 2- \kappa +  \min\lb \gamma,1\rb\rcb. 
	\end{align}
\end{theorem}
\begin{proof}
	See Appendix~\ref{sec:securegdof}.
\end{proof}
\textit{Remarks:} 
\begin{enumerate}
	\item The outer bound on the sum rate in Theorems~$4$  and $5$ are used to obtain outer bound on the sum GDOF~\cite{bagheri-arxiv-2010, partha-outer-arxiv-2016, partha-outer-submitted-2016}. Both the bounds give the same results in terms of GDOF. Note that the derivation of the outer bound in Theorem~$4$ does not use the secrecy constraint at receiver~$1$ \cite{bagheri-arxiv-2010}. Hence, there is no penalty in the sum GDOF due to the secrecy constraint at receiver in the weak/moderate interference regime for all values of $C_G$. 
	\item When $\gamma = \kappa$, $d_{\text{sum}}(\kappa, \gamma) = 2 $. Hence, both the users can achieve the maximum GDOF of $1$ simultaneously. Similarly, in the deterministic model, when $C = n$ $(\text{or } \frac{C}{m} = \alpha)$, both the users can achieve a maximum rate of $m$ simultaneously. 
\end{enumerate}

As the proposed scheme with power allocation in (\ref{eq:gap2}) can achieve the optimal sum GDOF, the achievable sum rate should be within finite number of bits from the outer bound. In the following subsection, the gap between the achievable sum rate and outer bound is characterized. 
\subsection{Finite bit gap result on the sum rate capacity}
In this section, the sum rate capacity of the $2$-user Gaussian Z-IC with unidirectional transmitter cooperation
is shown to lie within $2$ bits/s/Hz of the outer bound in the weak/moderate interference regime $(\INRt < \SNRt)$ for all values of $C_G$. Note that this gap is the worst case gap. To show the finite gap result, the power allocation in (\ref{eq:gap2}) is used in Corollary~\ref{th:corr-ib-weak} to obtain a lower bound on the secure sum rate capacity. This result is stated in the following theorem.
\begin{theorem}\label{th:finitegapresult}
	The secure sum rate capacity of the $2$-user Gaussian Z-IC with unidirectional transmitter cooperation lies within $2$ bits/s/Hz of the sum rate outer bound in the weak/moderate interference regime for all values of $C_G$, i.e., 
	\begin{align}
	R_{\text{sum}} \leq C_{\text{sum}} \leq C_{\text{sum}}^{\text{outer}} \leq R_{\text{sum}} + 2, \label{eq:gapresult}
	\end{align} 
	where $R_{\text{sum}}$ and $C_{\text{sum}}^{\text{outer}}$ correspond to the lower bound and upper bound on the secure sum capacity (i.e., $C_{\text{sum}}$), respectively. 
\end{theorem}
\begin{proof}
	See Appendix~\ref{sec:finitebitresult}.
\end{proof}
\section{Numerical Results and Discussion}\label{sec:results}
In the following sections, some numerical examples are presented for the deterministic and Gaussian cases,
to get insights into the system performance in different interference regimes.
\subsection{Deterministic Z-IC with unidirectional transmitter cooperation}
In Fig.~\ref{fig:capacity_region_m5n3}, the capacity region given in Theorem~\ref{th:theorem-det-weak}
is plotted for $m=5$, $n=3$ and various values of $C$. 
When $C=0$, there are three interfering links $[m-n + 1: m]$ from transmitter~$2$ to 
receiver~$1$. These links act as a shared data channel, i.e., they are usable only by one of the two 
transmitters for data transmission. It is interesting to note that under the secrecy constraint at
receiver~$1$, user~$2$ can achieve $5$ bits per channel use (bpcu), when user~$1$ achieves $2$ bpcu. This is due to the fact that
transmitter~$1$ can send a jamming signal (random bits) on the levels $[1:n]$, which enables user~$2$ to send data bits securely on the 
higher levels $[m-n+1:m]$ in addition to transmitting on the lower levels $[1:m-n]$. Recall that the links corresponding to the lower levels $[1:m-n]$ do not exist at receiver~$1$. Hence, user~$1$ has to sacrifice a rate of $n=3$ bpcu, which can also be observed
from the sum rate constraint in Theorem~\ref{th:theorem-det-weak}.  For $C \geq 3$, the capacity region 
becomes the square region, where 
both users can simultaneously achieve $m=5$ bpcu.  Note that the outer bound used to establish the capacity region in Theorem~\ref{th:theorem-det-weak} does not use the secrecy constraints at the receivers. Therefore, Fig.~\ref{fig:capacity_region_m5n3} is also a plot of the 
capacity region without secrecy constraints. 
\begin{figure}[t]
	\begin{center}
		\includegraphics[width=3.5in, height=3in]{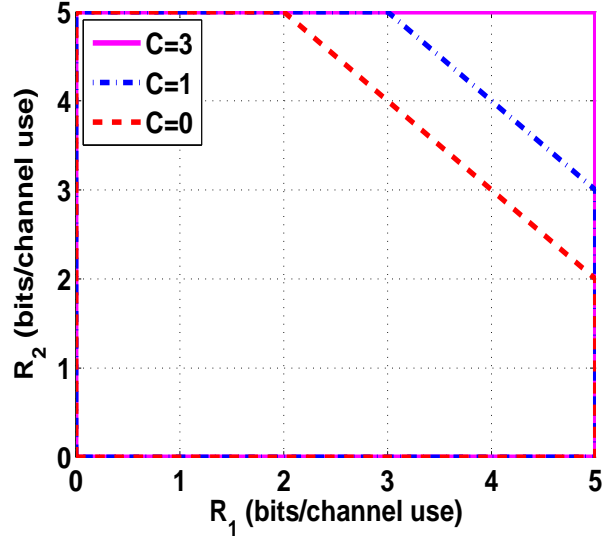}
		\caption{Secrecy capacity region of the deterministic Z-IC with~$(m,n)=(5,3)$. This corresponds to the weak interference regime.}
		\label{fig:capacity_region_m5n3}
	\end{center}
\end{figure}
\begin{figure}[t]
	\begin{center}
		\includegraphics[width=3.5in, height=3in]{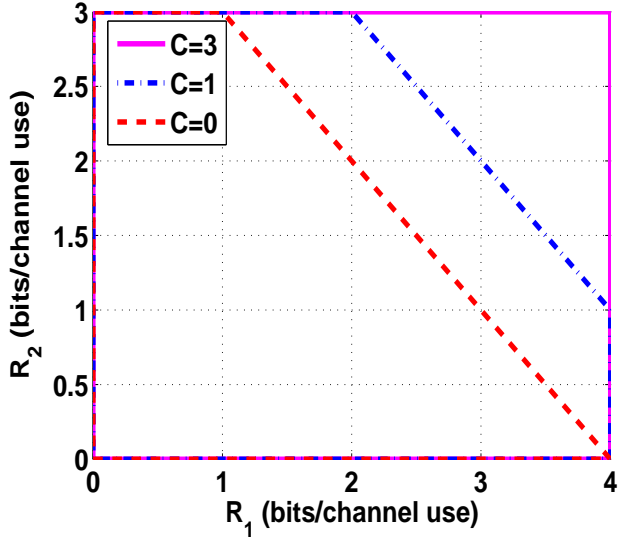}
		\caption{Secrecy capacity region of the deterministic Z-IC with~$(m,n)=(4,5)$. This corresponds to the high interference regime.}\vspace{-0.7cm}
		\label{capacity_region_m4n5}
	\end{center}
\end{figure}
In Fig.~\ref{capacity_region_m4n5}, the capacity region given in Theorem~\ref{th:theorem-det-high} is 
plotted for $m=4$, $n=5$ and various values of $C$. When 
$C=0$ and user~$1$ achieves the maximum rate of $m$, i.e., $4$ bpcu, user~$2$ cannot achieve any nonzero 
secrecy rate. When user~$2$ achieves the maximum rate of $2m-n$, i.e., $3$ bpcu, user~$1$ achieves a nonzero 
secrecy rate of $1$ bpcu. When $C=1$ and user~$1$ achieves the maximum rate of $4$ bpcu, from the sum rate
constraint in Theorem~\ref{th:theorem-det-high}, user~$2$ can achieve a rate of $1$ 
bpcu. On the other hand, when user~$2$ achieves a rate of $3$ bpcu, user~$1$ 
can achieve a rate of $2$ bpcu. Note that both the users can achieve a nonzero secrecy rate with cooperation, 
in contrast to the non-cooperating case.  The improvement in the capacity region 
with increase in value of $C$ is clear from the figure. 

In Figs.~\ref{fig:capacity_region_m5n3} and \ref{capacity_region_m4n5}, when $C > 0$, a part of the interference can be canceled at the unintended receiver, which 
leads to a gain in the sum rate. The need for sending a jamming signal also decreases 
with increase in the value of $C$. Hence, as  $C$ increases, the 
capacity region enlarges.

In Fig.~\ref{fig:sum_capacity_deterministic}, the secure sum capacity of the deterministic Z-IC is plotted against $\alpha$ for different values of $C$ using the result in Sec.~\ref{sec:ach-determin}. In this case, the secure sum capacity is normalized by $m$. When $C=0$, as $\alpha$ increases, the sum capacity decreases and becomes constant for $\alpha > 1$.  As the value of the cooperative link increases, in the initial part of the weak interference regime, both the users can achieve the maximum rate, i.e., $m$. This is due to the fact that capacity of the cooperative link is sufficient to cancel the interference at receiver~$1$. However, with further increase in the value of $C$, the secure sum capacity starts decreasing. In the very high interference regime, user~$2$ cannot achieve any nonzero secrecy rate irrespective of the value of $C$.
\begin{figure}[t]
	\begin{center}
		\includegraphics[width=3.5in, height=3in]{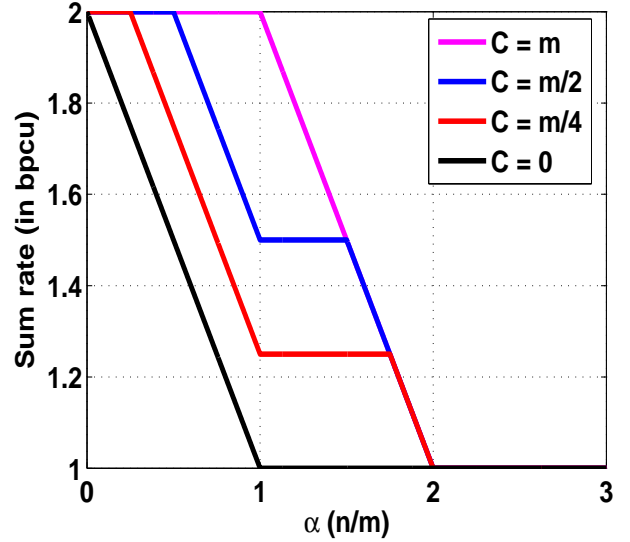}
		\caption{Secure sum capacity for the deterministic Z-IC. Rate is normalized with $m$.}
		\label{fig:sum_capacity_deterministic}
	\end{center}
\end{figure}
\subsection{Gaussian Z-IC with unidirectional transmitter cooperation}
In Figs.~\ref{fig:gaussian_result1} and \ref{fig:gaussian_result2}, the achievable results in Corollary~\ref{th:corr-ib-weak} are plotted along with
the outer bounds obtained in \cite{partha-outer-arxiv-2016, partha-outer-submitted-2016} for different values of $C_G$, in the weak and high
interference regime, respectively. When $C_G > 0$, a part of the interference can be canceled at the unintended receiver, which leads to a gain in the rate due to cooperation. In particular, the improvement in the sum rate performance for both the cases can be observed from these figures. As the capacity of the cooperative link increases, less power is assigned to send the non-cooperative private message of transmitter~$2$, which in turn also reduces the loss in rate due to stochastic encoding.
\begin{figure}
	\centering
	\includegraphics[width=3.5in, height=3in]{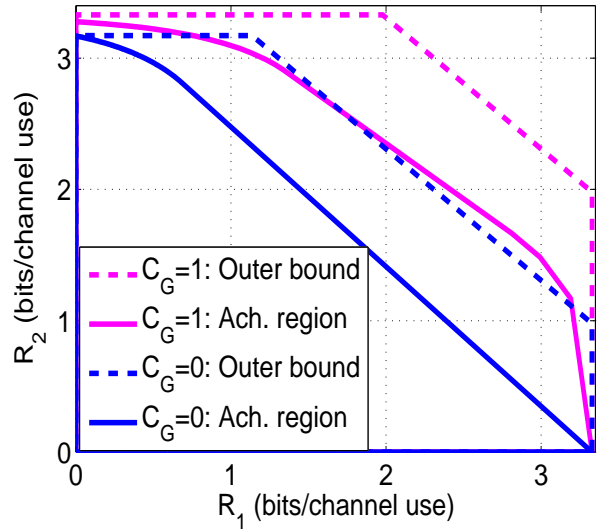}\\
	\caption{Achievable rate region for the Gaussian model: $P=100$, $h_d=1$ and $h_c=0.5$.}\label{fig:gaussian_result1}
\end{figure}
\begin{figure}
	\centering 
	\includegraphics[width=3.5in, height=3in]{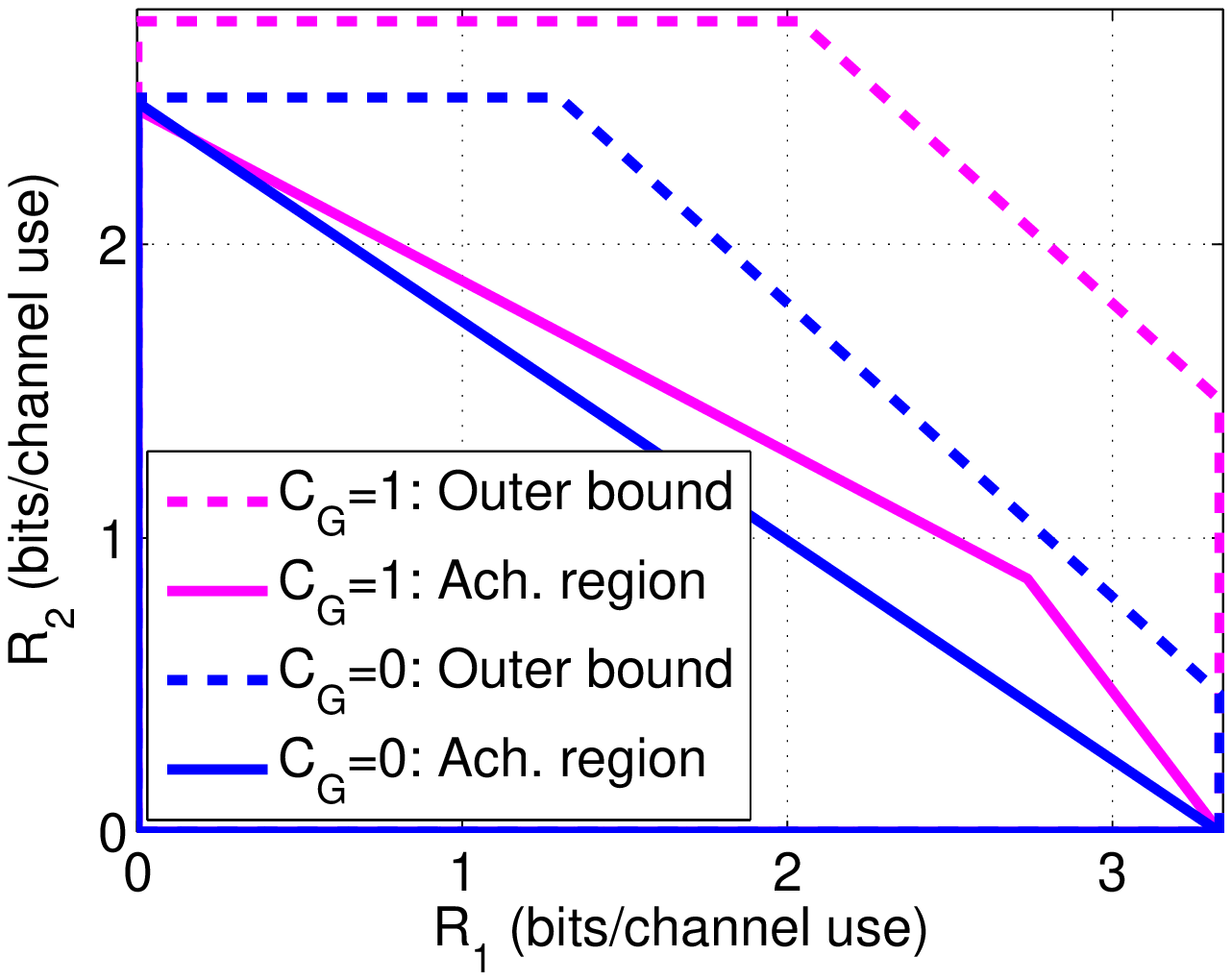}\\
	\caption{Achievable rate region for the Gaussian model: $P=100$, $h_d=1$ and $h_c=1.5$.}\label{fig:gaussian_result2}
\end{figure}
\begin{figure}
	\centering
	\includegraphics[width=3.5in, height=3in]{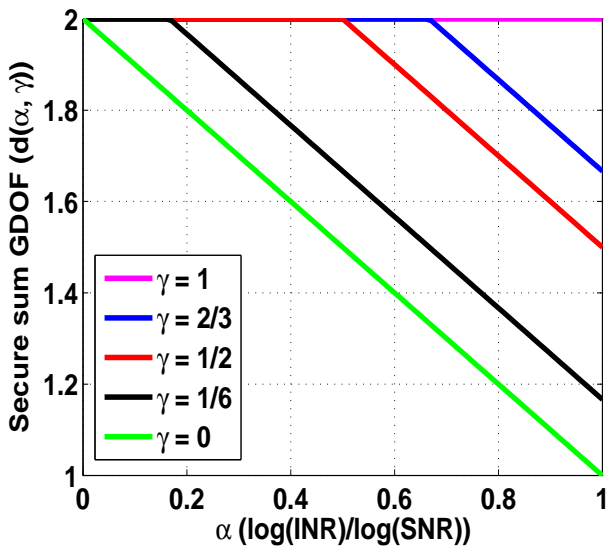}\\
	\caption{Secure sum GDOF in the weak/moderate interference regime for the Gaussian model. In the plot, $\gamma$ corresponds to the scaling of the capacity of the cooperative link with respect to $0.5\log \SNRt$.}\label{fig:gdof_result1}
\end{figure}

In Fig.~\ref{fig:gdof_result1}, the secure sum GDOF stated in Theorem~\ref{th:sumgdof} is plotted against $\alpha$ for various value of $\gamma$. From the figure, it can be noticed that with cooperation it is possible for both the users to achieve the maximum GDOF, i.e., $1$, in the initial part of the weak/moderate interference regime, if the capacity of the cooperative link scales with SNR. In these cases, there is no loss in terms of GDOF due to the secrecy constraint at the receiver. 

\section{Conclusions}\label{sec:conclusion}
This work explored the role of limited-rate unidirectional transmitter
cooperation in facilitating secure communication over the $2$-user Z-IC. For the
deterministic case, the achievable scheme uses a combination of interference
cancelation and transmission of random bits. The secrecy capacity region of the deterministic model is characterized over all interference regimes and for all values of $C$. The study of the deterministic model gave useful insights for
the Gaussian case. The proposed scheme for the Gaussian model uses a fusion
of cooperative precoding for interference cancelation, stochastic encoding and
artificial noise transmission for ensuring secrecy of the unintended message at the receiver.  The secure sum GDOF of the Gaussian Z-IC is characterized for the weak/moderate interference regimes. The sum 
rate capacity is also shown to lie within~$2$~bits of the outer bound in
the weak/moderate interference regime for all values of the capacity of the cooperative link, $C_G$. It is found that cooperation between the users can facilitate secure communication over Z-IC except for the very high interference regime. It is also found that secrecy constraint at the receiver does not hurt the capacity in the weak/moderate interference regime for the deterministic model. Similarly, it is found that there is no loss in the secure sum GDOF in the weak/moderate interference regime due to the secrecy constraint at the receiver.
\appendix
\subsection{Proof of Theorem~\ref{th:theorem-ib-weak}}\label{sec:append-ZIC-inner}
The proof involves analyzing the error probability at the decoders for the proposed encoding scheme,  along 
with equivocation computation. The
equivocation computation is necessary to choose how much of its own rate transmitter~$2$ must sacrifice  to
keep the non-cooperative private message secret. One of the main novelties of the proof lies in precoding of
the cooperative private message of transmitter~$2$ at transmitter~$1$, which cancels the interference at receiver~$1$ and at the 
same time
ensures secrecy of the cooperative private message.
\subsubsection{Error probability analysis} For receivers~$1$ and $2$, define the following events
\begin{align}
E_i & \triangleq \{(\ybold_1^N, \xbold_{p1}^N(i)) \in T_\epsilon^N(P_{Y_1X_{p1}})\}, \label{eq:weakint7} \\
F_{ijk} & \triangleq \{(\ybold_2^N, \xbold_{p2}^N(i,j), \xbold_{cp2}^N(k)) \in T_\epsilon^N(P_{Y_2X_{p2}X_{cp2}})\}, \label{eq:weakint8}
\end{align}
where $T_\epsilon^N(P_{Y_1X_{p1}})$ denotes the set of jointly typical sequences
$\ybold_1$ and $\xbold_{p1}$ with respect to $P(\ybold_1,\xbold_{p1})$ and
$T_\epsilon^N(P_{Y_2X_{p2}X_{cp2}})$ denotes the set of jointly typical sequences $\ybold_2$,
$\xbold_{p2}$ and $\xbold_{cp2}$ with respect to $P(\ybold_2,\xbold_{p2},\xbold_{cp2})$. Without loss of
generality, assume that transmitters~$1$ and $2$ send $\xbold_1^N(1,1)$ and $\xbold_2^N(1,1,1)$, respectively.
An error occurs if the transmitted and received codewords are not jointly typical, or a wrong codeword is
jointly typical with the received codewords. Using the union of events bound and asymptotic equipartition property (AEP), it can be shown that 
$\lambda_{e1}^N  = P(E_1^c \bigcup \cup_{i \neq 1} E_i) \leq P(E_1^c) + \displaystyle\sum_{i \neq 1} P(E_i) \rightarrow 0$ as $N \rightarrow \infty$ provided
\begin{align}
R_1 \leq I(\xbold_{p1}; \ybold_1). \label{eq:weakint11}
\end{align}

Similarly, the probability of error at receiver~$2$, i.e., $\lambda_{e2}^N  = P(F_{111}^c \bigcup \cup_{(i, j, k) \neq (1, 1, 1)} F_{ijk}) \leq P(F_{111}^c) + \displaystyle\sum_{(i, j, k) \neq (1, 1, 1)} P(F_{ijk}) \rightarrow 0$ as $N \rightarrow \infty$ provided
\begin{align}
R_{p2} + R_{p2}' & \leq I(\xbold_{p2};\ybold_2|\xbold_{cp2}), \label{eq:weakint14} \\
R_{cp2} & \leq I(\xbold_{cp2};\ybold_2|\xbold_{p2}),  \label{eq:weakint16} \\
R_{p2} + R_{p2}' + R_{cp2} & \leq I(\xbold_{p2},\xbold_{cp2};\ybold_2). \label{eq:weakint18}
\end{align}
Due to rate-limited cooperation, the following condition is required to be satisfied for the cooperative private
message
\begin{align}
R_{cp2} \leq C_G. \label{eq:weakint19}
\end{align}
Hence, using (\ref{eq:weakint11}), (\ref{eq:weakint14}), (\ref{eq:weakint16}), (\ref{eq:weakint18}), (\ref{eq:weakint19}), and $R_2 = R_{p2} + R_{cp2}$, (\ref{eq:weakint6a}) is obtained. 

In the following, $R_{p2}'$ is determined for ensuring secrecy of the non-cooperative private message of transmitter~$2$ at receiver~$1$.
\subsubsection{Equivocation computation}
Writing $H(W_2|\ybold_1^N)  = H(W_{p2}|\ybold_1^N) + H(W_{cp2}|\ybold_1^N, W_{p2})$, and noting that $H(W_{cp2}|\ybold_1^N, W_{p2}) = H(W_{cp2})$ because the codeword carrying the cooperative private message is completely canceled at receiver~$1$
and the cooperative private message is chosen independent of the non-cooperative private message at transmitter~$2$. Hence, to ensure secrecy of the message of transmitter~$2$, it is required to show the following
\begin{align}
H(W_{p2}|\ybold_1^N) \geq NR_{p2} - N\epsilon_N. \label{eq:weakint22}
\end{align}

The term $H(W_{p2}|\ybold_1^N)$ can be lower bounded as follows
\begin{align}
& H(W_{p2}|\ybold_1^N)\nonumber \\
& \geq H(W_{p2}|\ybold_1^N, \xbold_{p1}^N), \nonumber \\
& \stackrel{(a)}{=}  H(W_{p2}) - I(W_{p2};\ybold_1^N|\xbold_{p1}^N), \nonumber \\
& \stackrel{(c)}{=}  H(W_{p2}) - I(W_{p2};\ybold_1^N|\xbold_{p1}^N) + I(W_{p2};\ybold_1^N|\xbold_{p1}^N, \xbold_{p2}^N), \nonumber \\
& = H(W_{p2}) - h(\ybold_1^N|\xbold_{p1}^N) + h(\ybold_1^N|\xbold_{p1}^N, W_{p2}) \nonumber \\
& \qquad + h(\ybold_1^N|\xbold_{p2}^N, \xbold_{p1}^N) - h(\ybold_1^N|\xbold_{p2}^N, \xbold_{p1}^N, W_{p2}), \nonumber \\
& = H(W_{p2}) - I(\xbold_{p2}^N; \ybold_1^N|\xbold_{p1}^N) + I(\xbold_{p2}^N; \ybold_1^N|\xbold_{p1}^N, W_{p2}), \nonumber \\
& \geq  H(W_{p2}) - NI(\xbold_{p2}; \ybold_1|\xbold_{p1}) \nonumber \\
& \qquad + I(\xbold_{p2}^N; \ybold_1^N|\xbold_{p1}^N, W_{p2})-N\epsilon_N, \label{eq:weakint23}
\end{align}
where (a) is obtained using the relation: $I(W_{p2};\ybold_1^N|\xbold_{p1}^N) = H(W_{p2}|\xbold_{p1}^N) - H(W_{p2}|\ybold_1^N, \xbold_{p1}^N)$ and
the fact that $W_{p2}$ is  independent of $\xbold_{p1}^N$; and (c) is obtained
using the fact that $W_{p2} \rightarrow (\xbold_{p1}, \xbold_{p2}) \rightarrow \ybold_1$ forms a Markov chain
and hence, $I(W_{p2};\ybold_1^N|\xbold_{p1}^N, \xbold_{p2}^N)=0$.

The third term in (\ref{eq:weakint23}) is simplified  as follows
\begin{align}
I(\xbold_{p2}^N; \ybold_1^N|\xbold_{p1}^N, W_{p2})
& = H(\xbold_{p2}^N| W_{p2}) - H(\xbold_{p2}^N|\ybold_1^N, \xbold_{p1}^N, W_{p2}), \nonumber \\
& = N R_{p2}' - H(\xbold_{p2}^N|\ybold_1^N, \xbold_{p1}^N, W_{p2}). \label{eq:weakint25}
\end{align}
To bound the entropy term in (\ref{eq:weakint25}), consider the decoding of $W_{p2}'$ at receiver~$1$
assuming that a genie has provided $W_{p2}$ as side-information to receiver~$1$. For a given
$W_{p2}=w_{p2}$, assuming that $w_{p2}'$ was sent by transmitter~$2$ and receiver~$1$ knows
$\ybold_1^N = y_1^N$ and $\xbold_{p1}^N = x_{p1}^N$, receiver~$1$ tries to decode $w_{p2}'$. Receiver~$1$ declares that $j$ was sent if $(\xbold_{p2}^N(w_{p2}, j), \ybold_1^N)$ is
jointly typical and such $j$ exists and is unique. Otherwise an error is declared. Using AEP, it can be shown
that the probability of error can be made arbitrarily small if
\begin{align}
R_{p2}' \leq I(\xbold_{p2};\ybold_1|\xbold_{p1}).  \label{eq:weakint26}
\end{align}
When the condition in (\ref{eq:weakint26}) is satisfied and for sufficiently large $N$, using Fano's inequality, the following is obtained
\begin{align}
H(\xbold_{p2}^N|\ybold_1^N, \xbold_{p1}^N, W_{p2}= w_{p2}) \leq \delta_1. \label{eq:weakint27}
\end{align}
Using the above, the last term in (\ref{eq:weakint25}) is bounded as follows:
\begin{align}
& H(\xbold_{p2}^N|\ybold_1^N, \xbold_{p1}^N, W_{p2}) \nonumber \\
&  = \displaystyle\sum_{w_{p2}} P(w_{p2}) H(\xbold_{p2}^N|\ybold_1^N, \xbold_{p1}^N, W_{p2}= w_{p2}), \nonumber \\
& \leq N\delta_1. \label{eq:weakint28}
\end{align}
Hence, given $W_{p2}$, the codeword chosen for the non-cooperative private message for transmitter~$2$ is a \textit{good code} for receiver~$1$ with high probability if the condition in (\ref{eq:weakint26}) is satisfied.

Using (\ref{eq:weakint25}) and (\ref{eq:weakint28}), (\ref{eq:weakint23}) reduces to the
following
\begin{align}
H(W_{p2}|\ybold_1^N) & \geq N R_{p2} - NI(\xbold_{p2};\ybold_1|\xbold_{p1}) + NR_{p2}' - N\epsilon_N.  \label{eq:weakint30}
\end{align}
By choosing $R_{p2}' = I(\xbold_{p2};\ybold_1|\xbold_{p1}) - \epsilon$, secrecy of the non-cooperative
private message is ensured and (\ref{eq:weakint22}) is obtained. Substituting this choice into (\ref{eq:weakint30})  leads to (\ref{eq:weakint6a}).
\subsection{Proof of Corollary~\ref{th:corr-ib-weak}}\label{sec:append-ZIC-inner-corr}
The first term in (\ref{eq:weakint6a}) is evaluated as follows
\begin{align}
R_1  \leq  0.5\log\left(1 + \frac{h_d^2P_{p1}}{1 + h_d^2P_{a1} + h_c^2 P_{p2}}\right)\label{eq:corrach2}
\end{align}
where the power allocations are as mentioned in the statement of the theorem. The second term in
(\ref{eq:weakint6a}) is simplified as follows
\begin{align}
R_2 &\leq 0.5\log(1 + h_d^2P_{p2} + h_d^4P_{cp2}) - R_{p2}', \label{eq:corrach3}
\end{align}
 where $R_{p2}' = 0.5\log\lb 1 + \frac{h_c^2P_{p2}}{1 + h_d^2P_{a1}}\rb$.
 
The last term in (\ref{eq:weakint6a}) is simplified as follows
\begin{align}
R_2 & \leq 0.5\log(1 + h_d^2P_{p2}) + \min \lcb C_G, 0.5\log(1 + h_d^4P_{cp2})\rcb \nonumber \\
& \qquad  - R_{p2}'. \label{eq:corrach5}
\end{align}
Taking convex closure of (\ref{eq:corrach2}) and the minimum of (\ref{eq:corrach3}) and  (\ref{eq:corrach5}) over
different values of $\theta_i, \beta_i$ and  $\lambda_i$, the achievable secrecy rate in (\ref{eq:corrach0}) is obtained. 
The parameters $\theta_i, \beta_i$ and  $\lambda_i$ are defined in the statement 
of the Corollary. This completes the proof. 
\subsection{Proof of Theorem~\ref{th:sumgdof}}\label{sec:securegdof}
Using Corollary~\ref{th:corr-ib-weak} and the power allocation in  (\ref{eq:gap2}), the lower bound on the sum rate reduces to
\begin{align}
& R_1 + R_2 \nonumber \\
& \leq 0.5\log \lb 1 + \frac{P}{4}\rb + \min\lcb 0.5\log\lb 1 + \frac{1}{2h_c^2} + \frac{P}{2}\rb, \right. \nonumber \\
& \left. 0.5\log \lb 1 + \frac{1}{h_c^2}\rb + \min 
\lcb C_G, 0.5\log \lb 1 + \frac{1}{2} \lb P - \frac{1}{h_c^2}\rb\rb\rcb\rcb \nonumber \\
&  - 0.5\log2, \nonumber \\
& = 0.5\log \SNRt + \min\lcb 0.5\log\SNRt, 0.5\log \frac{\SNRt}{\INRt} \right. \nonumber \\
& \left. \qquad + \min\lcb C_G, 0.5\log \SNRt\rcb\rcb + \oone, \nonumber \\
\text{or } & d_{\text{sum}}(\kappa, \gamma)  = \min\{2, 2 -\kappa + \min\lb 1, \gamma\rb\}. \label{eq:sumgdof2}
\end{align}
Hence, the achievable sum GDOF becomes
\begin{align}
d_{\text{sum}}(\kappa, \gamma) = \min\lcb 2, 2-\kappa + \gamma\rcb.    \label{eq:sumgdof4}
\end{align}
To establish the GDOF optimality of the proposed scheme, considers the following outer bounds on the sum rate. As the individual rates of each user is upper bounded by $0.5\log\lb 1 + \SNRt\rb$,  a trivial outer bound on the sum rate is: $R_1 + R_2 \leq \log\lb 1 + \SNRt\rb$. Hence, the outer bound on the secure sum GDOF becomes $d_{\text{sum}}(\kappa, \gamma) \leq 2$. 

Next, consider the outer bound on the sum rate in Theorem~$4$ in \cite{bagheri-arxiv-2010}
\begin{align}
& R_1 + R_2 \nonumber \\
& \leq 0.5\log(1 + \text{SNR} + \text{INR} + 2\sqrt{\text{SNR}\cdot\text{INR}}) \nonumber \\
& \quad + 0.5\log\lb1 + \frac{\text{SNR}}{1 + \text{INR}} \rb + C_G, \nonumber \\
& \leq 0.5\log\lb 1 + 3\SNRt + \INRt\rb + 0.5 \log\lb 1 + \SNRt + \INRt\rb\nonumber \\
&\qquad - 0.5\log\lb 1 + \INRt\rb + C_G, \nonumber \\
& = \log \SNRt - 0.5\log \INRt + C_G + \oone, \nonumber \\
\text{or } & d_{\text{sum}}(\kappa, \gamma)  \leq 2 - \kappa + \gamma.  \label{eq:sumgdof7}
\end{align}
Next, starting from the sum rate bound in Theorem~$5$ in \cite{partha-outer-arxiv-2016} and using a similar procedure as the above, it can be shown that $d_{\text{sum}}(\kappa, \gamma) \leq 2 - \kappa + \gamma$. Hence, although (unlike Theorem~$4$ in \cite{bagheri-arxiv-2010}) Theorem~$5$ in \cite{partha-outer-arxiv-2016} was derived accounting for the secrecy constraint, both the theorems lead to the same outer bound on the GDOF. Finally, the outer bound on the secure sum GDOF  becomes
\begin{align}
d_{\text{sum}}(\kappa, \gamma) & \leq \min\lcb 2, 2 - \kappa + \gamma \rcb. \label{eq:sumgdof9}
\end{align}
It is easy to verify that the outer bound on the GDOF in (\ref{eq:sumgdof9}) coincides with the achievable GDOF in (\ref{eq:sumgdof4}). Hence, the proposed scheme is GDOF optimal, and this completes the proof.
\subsection{Proof of Theorem~\ref{th:finitegapresult}}\label{sec:finitebitresult}
Using Corollary~\ref{th:corr-ib-weak} and the power allocation in (\ref{eq:gap2}), the lower bound on the sum rate reduces to
\begin{align}
& R_1 + R_2 \nonumber \\
& \geq 0.5\log\left(1 + \frac{P_{p1}}{1 + h_c^2 P_{p2}}\right) + \text{min}\left\{\underbrace{0.5\log(1 + P_{p2} + P_{cp2})}_{I_1}, \right. \nonumber \\
& \left. \qquad \underbrace{0.5\log(1 + P_{p2}) + \min\{C_G, 0.5\log(1 + P_{cp2})\}}_{I_2} \right\} \nonumber \\
& \quad - 0.5\log(1 + h_c^2P_{p2}).\label{eq:gap1}
\end{align}
To determine the gap, the following exhaustive cases are considered. 
\subsubsection{When $I_1 \leq I_2$} In this case, (\ref{eq:gap1}) reduces to
\begin{align}
& R_1 + R_2 \nonumber \\
& \geq 0.5\log\lb 1 + \frac{P}{4}\rb + 0.5\log\lb 1 + \frac{1}{2h_c^2} + \frac{P}{2}\rb - 0.5\log2, \label{eq:gap3a} \\
& > 0.5\log(1 + \SNRt) + 0.5\log(1 + \SNRt) - 2. \label{eq:gap3}
\end{align}
A trivial outer bound on the sum rate is $R_1 + R_2 < \log(1 + \SNRt)$. Hence, comparing this outer bound on the sum rate with (\ref{eq:gap3}), the gap is at most $2$ bits/s/Hz.
\subsubsection{When $I_1 > I_2$ and $0.5\log(1 + P_{cp2}) > C_G$} In this case, the lower bound on the
sum rate in (\ref{eq:gap1}) reduces to 
\begin{align}
R_1 + R_2 & \geq 0.5\log\lb 1 + \frac{\SNRt}{4}\rb + 0.5\log\lb 1 + \frac{\SNRt}{\INRt}\rb + C_G \nonumber \\
& \quad -  0.5\log(1 + h_c^2P_{p2}), \nonumber \\
&  > 0.5\log(1 + \SNRt) + 0.5\log\lb 1 + \frac{\SNRt}{\INRt}\rb + C_G \nonumber \\
& \quad - 1.5. \label{eq:gap5}
\end{align}
To calculate the gap, the  following outer bound on the sum rate in Theorem~$5$ is used \cite{partha-outer-arxiv-2016}.
\begin{align}
R_1 + R_2 \leq \log\lb 1 + \SNRt\rb - 0.5\log\lb 1 + \INRt\rb + C_G. \label{eq:gap5a}
\end{align}
Subtracting (\ref{eq:gap5}) from the sum rate outer bound in (\ref{eq:gap5a}), 
it can be seen that the gap is at most $2$ bits/s/Hz. 
\subsubsection{When $I_1 > I_2$ and $0.5\log(1 + P_{cp2}) \leq C_G$} In this case, the lower bound on the
sum rate reduces to (\ref{eq:gap3a}), for which the gap is shown to be at most $2$ bits/s/Hz. 

Hence, the sum rate capacity of the Z-IC with unidirectional transmitter cooperation and the secrecy constraints at the
receivers is within $2$ bits/s/Hz of the outer bound. This completes the proof. 


\bibliographystyle{IEEEtran}
\bibliography{IEEEabrv,refs_onesidedIC}
\end{document}